\documentclass[a4paper,onecolumn,11pt,accepted=2025-12-04]{quantumarticle}
\pdfoutput=1
\usepackage[utf8]{inputenc}
\usepackage{hyperref}
\usepackage[british]{babel}
\usepackage[numbers,sort&compress]{natbib}

\usepackage{subcaption}
\usepackage[inline]{enumitem}
\setlist{nosep} \usepackage[normalem]{ulem}

\usepackage{physics}
\usepackage{qcircuit}
\usepackage[operators,sets]{cryptocode}
\usepackage{bbm}
\usepackage{amsmath}
\usepackage{amssymb}
\usepackage{mathtools}

\usepackage{amsthm}
\newtheorem{theorem}{Theorem}
\newtheorem{lemma}{Lemma}
\newtheorem{definition}{Definition}

\usepackage{algorithm}
\usepackage{algorithmic}
\floatstyle{ruled}

\newfloat{protocol}{htb!}{idf}
\floatname{protocol}{Protocol}

\newfloat{resource}{htb!}{idf}
\floatname{resource}{Resource}

\newfloat{simulator}{htb!}{idf}
\floatname{simulator}{Simulator}

\newcommand{\GHZ}{\textrm{GHZ}}
\newcommand{\sfrac}[2]{{}^{#1}\mathclap{/}_{#2}}
\newcommand{\Id}{\mathsf{I}}
\newcommand{\X}{\mathsf{X}}
\newcommand{\Y}{\mathsf{Y}}
\newcommand{\Z}{\mathsf{Z}}

\newcommand{\CZ}{\mathsf{CZ}}
\newcommand{\CNOT}{\mathsf{CNOT}}
\newcommand{\RZ}{\mathsf{R}_{\Z}}
\newcommand{\Ab}{\mathsf{Abort}}
\newcommand{\err}{\mathsf{Error}}

\title{Towards practical secure delegated quantum computing with semi-classical light}
\author{Boris Bourdoncle}
\affiliation{Quandela, 7 Rue Léonard de Vinci, 91300 Massy, France}
\author{Pierre-Emmanuel Emeriau}
\affiliation{Quandela, 7 Rue Léonard de Vinci, 91300 Massy, France}
\author{Paul Hilaire}
\affiliation{Quandela, 7 Rue Léonard de Vinci, 91300 Massy, France}
\author{Shane Mansfield}
\affiliation{Quandela, 7 Rue Léonard de Vinci, 91300 Massy, France}
\author{Luka Music}
\thanks{Corresponding author, luka.music@quandela.com}
\affiliation{Quandela, 7 Rue Léonard de Vinci, 91300 Massy, France}
\author{Stephen Wein}
\affiliation{Quandela, 7 Rue Léonard de Vinci, 91300 Massy, France}
\date{ }

\begin{document}

\begin{abstract}
Secure Delegated Quantum Computation (SDQC) protocols are a vital piece of the future quantum information processing global architecture since they allow end-users to perform their valuable computations on remote quantum servers without fear that a malicious quantum service provider or an eavesdropper might acquire some information about their data or algorithm. They also allow end-users to check that their computation has been performed as they have specified it.

However, existing protocols all have drawbacks that limit their usage in the real world. Most require the client to either operate  a single-qubit source or perform single-qubit measurements,
thus requiring them to still have some quantum technological capabilities albeit restricted, 
or require the server to perform operations which are hard to implement on real hardware (e.g isolate single photons from laser pulses and polarisation-preserving photon-number quantum non-demolition measurements).
Others remove the need for quantum communications entirely but this comes at a cost in terms of security guarantees and memory overhead on the server's side.

We present an SDQC protocol which drastically reduces the technological requirements of both the client and the server while providing information-theoretic composable security. More precisely, the client only manipulates an attenuated laser pulse, while the server only handles interacting quantum emitters with a structure capable of generating spin-photon entanglement. The quantum emitter acts as both a converter from coherent laser pulses to polarisation-encoded qubits and an entanglement generator. Such devices have recently been used to demonstrate the largest entangled photonic state to date, thus hinting at the readiness of our protocol for experimental implementations.
\end{abstract}
\maketitle

\section{Introduction}
\label{sec:intro}

\paragraph{Context.}
The number of potential use-cases for quantum computing steadily increases, whether in the noisy intermediate-scale quantum regime or using larger machines to perform fault-tolerant computations, with a diverse range of applications from pharmaceutics, material research, finance, to defence.

\begin{figure}[ht!]
    \centering
    \includegraphics[width=\columnwidth]{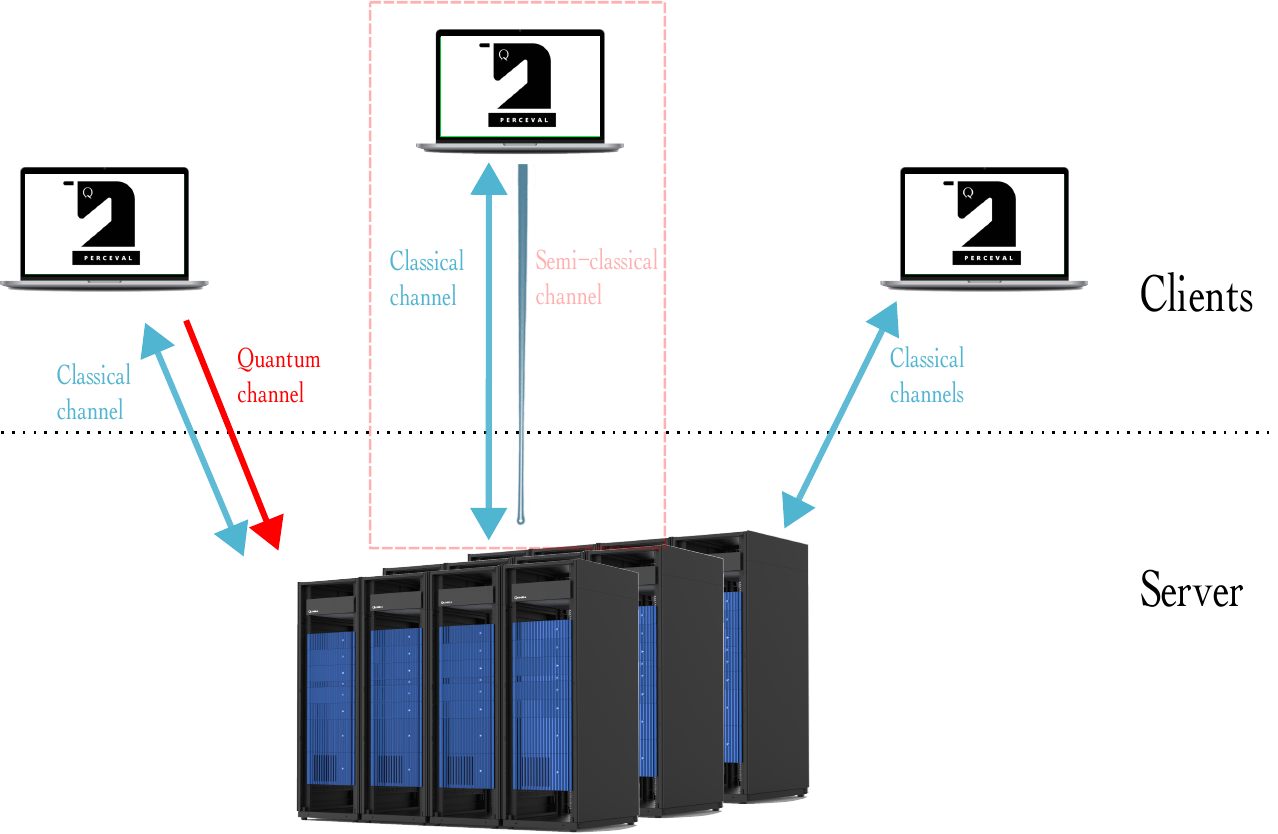}
    \caption{Typical client-server configuration where a limited client delegates a quantum computation to a quantum server. Clients might have a limited quantum capability e.g. the ability to send single qubits [left] or entirely classical channels [right]. In this work we focus on a hybrid proposal where clients have a semi-classical channels [centre] meaning that the client is able to send a coherent state.}
    \label{fig_client_server}
\end{figure}

However, most end-users will likely prefer not to own a quantum computer but still want to leverage its capabilities for their applications. The optimal strategy, therefore, is to use delegated  solutions where quantum computations are performed by a quantum service provider. This approach is actively being developed in both academic and industry settings.\footnote{Examples of public-accessible cloud platforms include:
Quandela \href{https://cloud.quandela.com}{https://cloud.quandela.com},
IBM \href{https://quantum-computing.ibm.com}{https://quantum-computing.ibm.com},
IonQ \href{https://ionq.com/quantum-cloud}{https://ionq.com/quantum-cloud}, 
Rigetti, \href{https://docs.rigetti.com/qcs/}{https://docs.rigetti.com/qcs/},
} 
This is commonly done via the combination of a code-based interface for describing the computation and classical communications between the client and the server.

Unfortunately, in this scenario there are no information-theoretic guarantees ruling out a malicious provider accessing all information about the client's desired computation, inputs, and outcomes.
Cases where such considerations could be critical include (i) the manipulation of highly sensitive input data, such as confidential medical records or defence-related information, (ii) the execution of proprietary (and confidential) algorithms, and (iii) cases where the output is valuable, e.g.~patent-worthy. 
Furthermore, the client cannot check in general that the service provider has performed its desired computation based only on the received output.
The more powerful quantum computers are, the more problematic these issues become, as they grow capable of solving more and more advanced problems.  All these considerations stiffen the adoption of quantum computers precisely at the time when they become useful for solving real-world problems. It is therefore vital to protect clients in these scenarii to foster wide-spread use of quantum computing technologies.

Thankfully, we can entirely remove the need for clients to trust their quantum service providers, with the help of quantum cryptography. Although quantum computers will be a threat to the security of some classical protocols (e.g. RSA)~\cite{Shor1999}, they also open up new possibilities via cryptographic protocols such as quantum key distribution~\cite{Bennett2020}. In our particular case, protocols for Secure Delegated Quantum Computing (SDQC) have been a highly active research field in the past few years and represent an extremely promising application~\cite{Childs2001, Arrighi2006, Broadbent2010, Fitzsimons2017, Gheorghiu2019}. These protocols allow a client with limited quantum resources to delegate a quantum computation to a server, without the server learning anything about the quantum algorithm, the input which is being processed or the outcome of the computation. During the protocol's execution, the client can furthermore test that the server is behaving as expected.

\paragraph{Related work.}
Whether SDQC protocols see wide-spread use will depend directly on how technically accessible they are from the point of view of both the client and the server. In this respect, there are two extremes. On one end of the spectrum lie protocols which require one-way qubit communications between the client and the server and rather involved quantum setups for the client -- either a source of qubits~\cite{Broadbent2010,Leichtle2021} or a way to measure those received from the server~\cite{Morimae2014,Greganti2016}, along with single qubit operations. These protocols are information-theoretically secure, meaning that their security relies on no assumptions beyond the laws of physics, and they leak no information whatsoever about the computation, input or output. They are also composable, i.e.~their security is not affected by the context of execution and they can be freely reused as black boxes in larger constructions (e.g. secure multi-party computations).
However, for many interesting use cases it may not be feasible to expect clients to have access to either a qubit source or measurement devices in the near term and so it is important to reduce the technological requirements on the client's side.

At the other end there are protocols~\cite{Cojocaru2019, Cojocaru2021, Gheorghiu2019, Mahadev2018} which only require the client to interact classically with the server instead of sending quantum states. However, they incur a significant overhead on the server's side, since it must be able to execute large cryptographically-secure functions to generate encrypted qubits of which only the client knows the secret key. Furthermore, these protocols only provide computational security, meaning that it relies on the quantum hardness of a computational problem. Finding an efficient algorithm for said problem would automatically break the protocol's security. It has further been proven that they cannot be composably secure without setup assumptions~\cite{Badertscher2020}.

The middle ground is represented by the protocol from~\cite{Dunjko2012}, whose efficiency in terms of communication was improved in~\cite{Zhao2017, Jiang2019}. The client here only manipulates and transmits semi-classical states, e.g. laser light for photonic applications, which drastically reduces the client's technological requirements. However, not only is it only blind and does not allow the client to verify the server's operation, there are three main challenges which limit its application. First of all, the server needs an efficient polarisation-preserving\footnote{Or qubit-preserving if other encoding are used.} photon-number Quantum Non-Demolition (QND) measurement, a very demanding operation in practice. Moreover, it must be able to extract from each pulse one single photon to use in the rest of the protocol. Finally, the server needs to implement deterministic two-qubit photonic gates, which are also extremely demanding. We also note that the security is information-theoretic but not proven in a composable security framework in the original paper.

\paragraph{Overview of results.}
Here, we propose a best-of-both-worlds protocol for SDQC using semi-classical communications that drastically reduces the technological requirements of \emph{both} the client and the server. In particular, the client does not need to manipulate qubits, only a phase-randomised attenuated laser source and active waveplates, while the server does not need photon-number QND measurements, nor deterministic photon-photon gates that were required in Ref.~\cite{Dunjko2012}. Furthermore, we prove that the security of our protocol is information-theoretic and composable in the Abstract Cryptography framework of \cite{Maurer2011}. Figure~\ref{fig_ubqc_proto} shows the various technological requirements of \cite{Broadbent2010, Dunjko2012} and our protocol.

\begin{figure}[htp]
    \centering
    \includegraphics[width=\columnwidth]{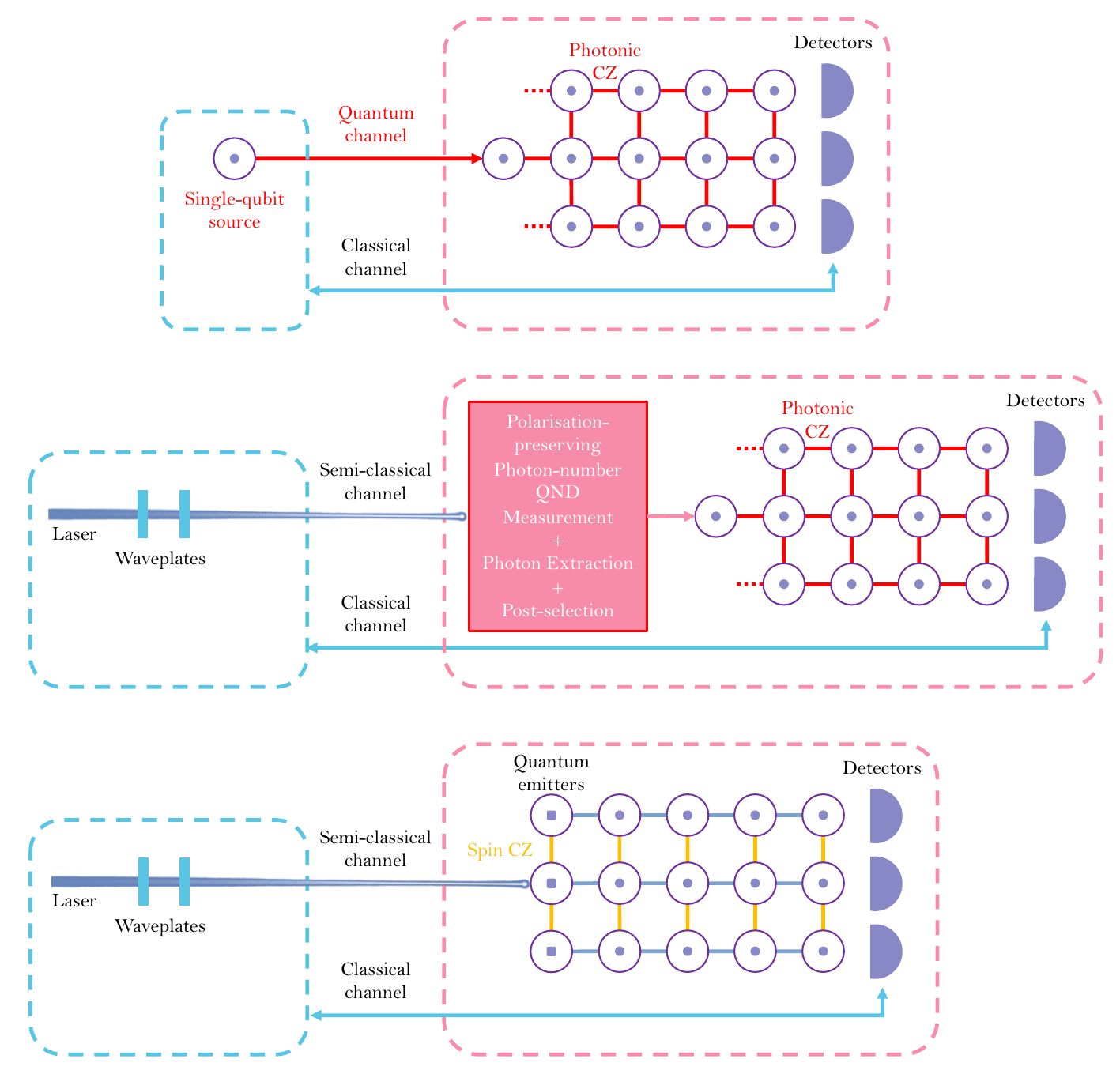}
    \caption{Hardware setup for performing three delegated computation protocols. The operations in the blue dashed boxes are performed by the client while the server handles those in the pink dashed boxes.
Red, yellow, and blue elements corresponds to operations that have a high, moderate, and low  technological requirements respectively. 
    In all three cases the client classically sends measurement instructions to the server and recovers the outcome. Top: UBQC protocol from Ref.~\cite{Broadbent2010}. Middle: BDQC protocol from Ref.~\cite{Dunjko2012}. Bottom: our SDQC protocol.}
    \label{fig_ubqc_proto}
\end{figure}\

Our protocol relies on Measurement-Based Quantum Computing (MBQC), a paradigm in which a large entangled state called graph state is first created and the computation is performed by adaptively measuring parts of the state in a predefined order, consuming the entanglement in order to carry out the computation. Such computations can be performed blindly by having the client prepare encrypted single qubits and send them to the server who then entangles them. The entanglement operation commutes with the encryption, thus yielding an encrypted graph state. The client instructs the server to measure qubits from the encrypted graph state in a way which both undoes the encryption and performs the computation of the client's choice. This is the basis of the original Universal Blind Quantum Computation (UBQC) protocol of~\cite{Broadbent2010}. The first challenge is therefore to emulate these single-qubit transmissions using only attenuated laser pulses. The second challenge is to reduce the number of entangling gates that need to be performed after the qubits have been generated.

Protocol~\ref{prot:ghz-gadget-inf} presents informally our technique for generating a single qubit with amplified privacy which we describe in more details below.

\begin{protocol}[htp]
\caption{GHZ privacy amplification for rotated states from weak coherent pulses (informal)}
\label{prot:ghz-gadget-inf}

\begin{enumerate}
\item The server initialises its quantum emitter in the state $\ket{+}$.
\item The client sends $n$ randomly-polarised weak coherent pulses, which the server uses to excite its quantum emitter. The energy level structure of the quantum emitter induces as output a polarisation-encoded GHZ state that is $\Z$-rotated by an angle corresponding to the sum of all random angles chosen by the client.
\item 	
\begin{enumerate}
\item The server emits one photon using a non-rotated coherent pulse and attempts to measure the first $n$ photonic qubits (generated by the client's pulses) in the $\ket{\pm}$ basis.
\item If the number of unsuccessful measurements (i.e.\ no photon was detected) is too high, the client aborts. Otherwise the server applies a correction to the last remaining photon and the quantum emitter which depends on the client's choice of angles and the measurement results. These corrections are specified in Protocol~\ref{prot:ghz-gadget}.
\end{enumerate}
\end{enumerate}

\end{protocol}

As expressed above, the client will send a sequence of phase-randomised attenuated laser pulses, with the polarisation of each pulse being chosen at random. This polarisation angle corresponds to the secret key which will later be used to encrypt the computation. The server receives these pulses and uses them to excite a quantum emitter with a qubit degree of freedom, such as a spin $\sfrac12$, that can emit spin-entangled photons.
The quantum emitter will have two roles. It will first act as a ``laser-to-qubit converter'' by converting a coherent pulse into a single photonic qubit. This role is somewhat similar to the Remote Blind Qubit State Preparation that was proposed in Ref.~\cite{Dunjko2012}, but we execute it without having to perform a qubit-preserving photon-number QND measurement. Moreover, it will act as a ``photon entangler'' as the emitted photonic qubit will be entangled with the spin. The server exploits this property to produce an entangled state of photonic qubits, removing the need of performing demanding photon-photon entangling gates to produce the encrypted graph state.
More precisely, it is possible to produce photonic linear cluster states with a single quantum emitter, with each photon being entangled in polarisation to the next emitted photon in the chain and the last one being also entangled to the spin \cite{Schon2005,Lindner2009,Schwartz2016, Coste2023, Cogan2023, Thomas2022}. Crucially, the state of each photonic qubit thus created inherits the polarisation angle from the pulse that was used to generate it. These chains of encrypted photonic polarisation-encoded qubits can then be combined to form the graph state supporting the client's desired computation. Only the links between chains need to be done after the qubits have been emitted, which drastically reduces the amount of entanglement operations. These can furthermore be done via interacting quantum emitters to remove entirely the need for photon-photon entangling gates, by performing the entanglement operations between quantum emitters in between photon emissions.

However, as in~\cite{Dunjko2012}, an attenuated laser pulse necessarily leaks more information about the client's secret key compared to a single qubit. This is due to the fact that the state of the laser pulse can be described as a redundant encoding of the single qubit state containing the secret key. Taking a simplifying assumption, we can overestimate the information leakage by considering that any multi-photon component in the input laser completely leaks the polarisation information.  We therefore need to suppress this information leakage in a way which meshes nicely together with the qubit generation technique described above. Fortunately, it is possible to easily switch between generating chains of qubits and GHZ states simply by selectively applying a Hadamard operation on the quantum spin. The full protocol then generates chains of encrypted GHZ states instead of chains of single qubits. Each qubit in the GHZ states is generated with a different random polarisation angle. The result is a chain of encrypted GHZ states whose encryption key is the sum of the polarisation angles. Each GHZ state can easily be collapsed into a single qubit while preserving the encryption key via measurements on all qubits but one in the GHZ state in the $\ket{\pm}$ basis, thus recovering the chain of single qubits used for the rest of the blind protocol. 
Intuitively, the server can learn the final encryption key of a given qubit only if all the angles used to generate the corresponding GHZ state have leaked. Therefore the final leakage probability is suppressed exponentially in the number of qubits per GHZ state.

This protocol exploits low-intensity lasers $\ket{\alpha}$ which generally have a non-negligible vacuum component, $|\bra{0}\ket{\alpha}|^2$. In that case, even with an honest server and without other sources of photon loss before or after emission, the quantum emitter sometimes cannot emit a photonic qubit. Our protocol accounts for this scenario by introducing a threshold on the number of detected photons during the aforementioned collapse of the GHZ states.
If the losses reported by the server after emission are too high, the client simply aborts.

The end result of the protocol in this simplified example is the state $\frac{1}{\sqrt{2}}(\ket{00} + e^{i\theta}\ket{11})$ for a random angle $\theta$ known only to the client. The first qubit corresponds to the quantum emitter's qubit degree of freedom, while the second is the state of the remaining polarisation-encoded photon. This gives a lot of flexibility in how the server can use this protocol to construct a larger graph state since the emitter and photon are still entangled. The server can easily extend this into a linear cluster state by repeating the protocol, or continue growing this rotated GHZ state so that it has multiple photons available for performing graph fusion operations.

We show that both the protocol's security error and correctness error decrease exponentially with the number of pulses sent by the client (stated formally in Th.~\ref{thm:sec-gadget} and proven in App.~\ref{app:sec-proof}).  

\begin{theorem}[Blind state preparation from weak coherent pulses, informal]
The probability of obtaining an incorrect state and the probability that information about the angle $\theta$ leaks to the server both exponentially decrease in the number of pulses $n$, as proven in a composable security framework.

\end{theorem}

When our protocol is used to replace the source of qubits required to perform the UBQC protocol from~\cite{Broadbent2010}, it yields the following result.

\begin{theorem}[Blind Delegated Quantum Computation (BDQC) from weak coherent pulses, informal]

There exists an efficient composably secure BDQC protocol for classical inputs with an exponentially-low security error in which the client only sends weak coherent pulses and the server only manipulates interacting quantum emitters.

\end{theorem}

We then use this to build a protocol for SDQC. In addition to BDQC's blindness, SDQC also guarantees that the server performs the client's computation as instructed. This means that any deviation from the supplied instructions which could corrupt the outcome of the computation should be detected with high probability. The tests which the client uses to detect these deviations up until recently all required the client to produce not only states in the $\X - \Y$ plane as in the UBQC protocol, but also computational basis states. This second requirement has been lifted recently in~\cite{Kapourniotis2023}, which constructs an SDQC protocol for classical input and output $\mathsf{BQP}$ computations from the same set of states as the one whose security we amplify here. Combining the result above with their protocol yields the following theorem.

\begin{theorem}[Secure Delegated Quantum Computation from weak coherent pulses, informal]

There exists an efficient composably secure SDQC protocol for delegating classical input and output $\mathsf{BQP}$ computations with an exponentially-low security error in which the client only sends weak coherent pulses and the server only manipulates interacting quantum emitters.

\end{theorem}

Finally, we study the experimental feasibility of our protocol by computing the trade-off between correctness and security in the presence of losses after emission. More precisely, we analyse the emission probability of a two-level quantum emitter with temporal laser filtering. We find that the single-photon emission efficiency can reach $0.71$ for a laser pulse containing an average of $2.5$ photons, which would be to amplify the correctness and security of our protocol. We also describe how other energy level structures of the quantum emitter could boost this, allowing us to work with even lower laser powers, thus improving our bounds. Recent experimental developments in that regard show extremely promising results~\cite{Schwartz2016, Cogan2023, Yang2022, Thomas2022, Coste2023} and hold the current record for the largest entangled photonic state produced~\cite{Thomas2022}. Our scheme is thus not only secure but also experiment-ready.

\paragraph{Comparison with Dunjko et al.~\cite{Dunjko2012}.}
The protocol from~\cite{Dunjko2012} works by having the server non-destructively measure the client's coherent light pulses in photon number, and post-select only the light pulses containing at least one photon. Then, the server must isolate one photon from each pulse. It applies a privacy amplification step which requires as many $\CZ$ gates as there are photons. After that, the server also has to entangle the privacy-enhanced photons into the graph state supporting the client's computation by using photon-photon $\CZ$ gates. All four of these steps are challenging technological requirements.

The Remote State Preparation (RSP) in Ref.~\cite{Dunjko2012} -- performing a polarisation-preserving QND measurement to count the photons in each pulse and extracting a single qubit from pulses containing at least one photon -- is separated from the privacy amplification and the graph state generation. This first step has been implemented by~\cite{Jiang2019}, requiring a number of pulses from the client which is around $10^8$ per qubit in the graph in the best case scenario where there is virtually no spatial separation between the client and server. They solve the issue of performing a QND measurement and qubit extraction -- which previously required a much more complex process~\cite{Guerlin2007, Reiserer2013} -- by using a linear optical setup performing teleportation with EPR pairs. Although this is rather fast -- their experiment generates $10^{11}$ leaky qubits via this RSP technique in two hours -- they still need to entangle these polarisation-encoded base qubits using photonic gates to perform the privacy amplification, which is not done in~\cite{Jiang2019}. Since they claim to want to use these $10^{11}$ leaky qubits to generate $10^3$ leak-free qubits, this implies a multiplicative overhead of $10^8$ in terms of entangling photonic gates compared to the client's base computation. The only viable option in this setting is to use deterministic non-post-selected photon-photon entangling gates since anything else will degrade far too rapidly to be useful. Unfortunately, the fidelities of such gates which have been demonstrated~\cite{Sun2016, Hacker2016, Reuer2022} are still insufficient to apply such a number of operations while preserving the integrity of the quantum state, even with error-correction techniques.

On the other hand, our RSP technique using a quantum emitter is intrinsically intertwined with the privacy amplification and graph state generation steps: the entanglement required for these two processes is directly generated by the structure of the quantum emitter and no additional operations are required. The result is that the final qubits are positioned in the graph at the same time as they are produced with a polarisation that is unknown to the server. Our full protocol does not require any non-destructive photon measurements nor deterministic photon-photon gates. Instead, it requires quantum emitters with a spin degree of freedom, with a level structure that enables the emission of spin-entangled photons that are also efficient light sources. Single atoms and artificial atoms such as quantum dots or defect color centres inserted in optical microcavities or waveguides are natural choices for such spin-photon interfaces.

Our scheme can naturally generate linear cluster states from a single quantum emitter. In order to perform universal computations, multiple states of this kind must be entangled together. This can be done in one of two ways.

If multiple spins generate linear cluster states in parallel and these spins can interact with each other, then applying nearest-neighbour spin-spin $\CZ$ gates when the clusters should be linked is sufficient to obtain universality~\cite{Thomas2024}. This is rather demanding since it requires either independently controlling chains of single atoms or artificial atoms such as quantum dot molecules, or the availability of a spin register such as a \textsuperscript{13}C nuclear spin for color centres in diamond.

Otherwise, it is possible to implement these operations with linear optics, via fusion operations and delays. The server can leverage the flexibility in our RSP technique to generate additional photons in each GHZ state and use these photons to implement fusion gates between different linear clusters as in~\cite{Hilaire2022}, at the cost of an increased sensitivity to photon loss after emission of the probability of generating the target blind graph.

\paragraph{Comparison with Takeuchi et al.~\cite{Takeuchi2024}.}
The recent paper~\cite{Takeuchi2024} also proposed a protocol for information-theoretically secure SDQC using only semi-classical light. There are however significant differences both in terms of results and techniques. Their protocol is based on the SDQC protocol from~\cite{Morimae2020} and requires the client to send states sampled from the same set as in our construction. The privacy amplification technique is similar to ours and that of Dunjko et al. However, their construction requires the same complex operations as Dunjko et al as well: QND measurements and photon-photon entangling gates. As such they do not resolve the issues stemming from these operations and our protocol is therefore much more suitable for implementation. Furthermore their security is not proven in a composable framework such as the one used in the present paper, meaning that their security guarantees do not automatically hold if an outcome of their protocol is later reused in another cryptographic construction.

\paragraph{Organisation of the paper.}
We start by giving preliminaries on
\begin{enumerate*}[label=(\roman*)]
    \item the generation of photonic graph states, and \item delegated quantum computing protocols.
\end{enumerate*}
In Section \ref{sec:rsp-semi-cl}, we present our protocol for securely preparing encrypted graph states from semi-classical light and its application in blind and secure delegated protocols. Later in Section \ref{sec:loss-analysis} we analyse the success probability of our protocol on realistic hardware, taking into account multiple effects leading to photon loss after emission. Finally, we discuss the merits of and possible improvements to our protocol in Section \ref{sec:dicuss}.

\section{Preliminaries}
\label{sec:prelim}

Section \ref{subsec:graph-gen} presents how to generate graph state from quantum emitters and weak coherent pulses. Then Section \ref{subsec:ubqc} describes the UBQC protocol that allows a client to blindly delegate quantum computations to a server. Additional preliminaries can be found in Appendix \ref{app:prelim}, in particular Appendix \ref{app:ac} presents the composable Abstract Cryptography security framework from \cite{Maurer2011} in which we prove the security of our protocols. From now on, we call the client ``Alice" and the server ``Bob". 
We denote $\abs{S}$ the cardinality of set $S$.

\subsection{Graph state generation using quantum emitters}
\label{subsec:graph-gen}

We describe here how to generate a cluster state from coherent pulses and quantum emitters using the technique from~\cite{Russo2019}.

A graph state is a quantum state defined by a graph $G = (V, E)$ such that all vertices $v \in V$ are associated to qubits initialised in the $\ket{+}$ state and an edge $e = (v, w) \in E$ between two vertices corresponds to a $\CZ$ operation between the two qubits associated to vertices $v$ and $w$. Formally the graph state is thus given by:
\begin{equation}
    \ket G = \prod_{(v,w) \in E} \CZ_{v,w} \bigotimes_{v \in V} \ket{+}_v.
\end{equation} 
A cluster state is a graph state whose graph is a regular lattice. In Figure \ref{fig:cluster-state}, we for example display a 2D-cluster state.

\begin{figure}[ht]
\centering
\includegraphics[height=5.5cm]{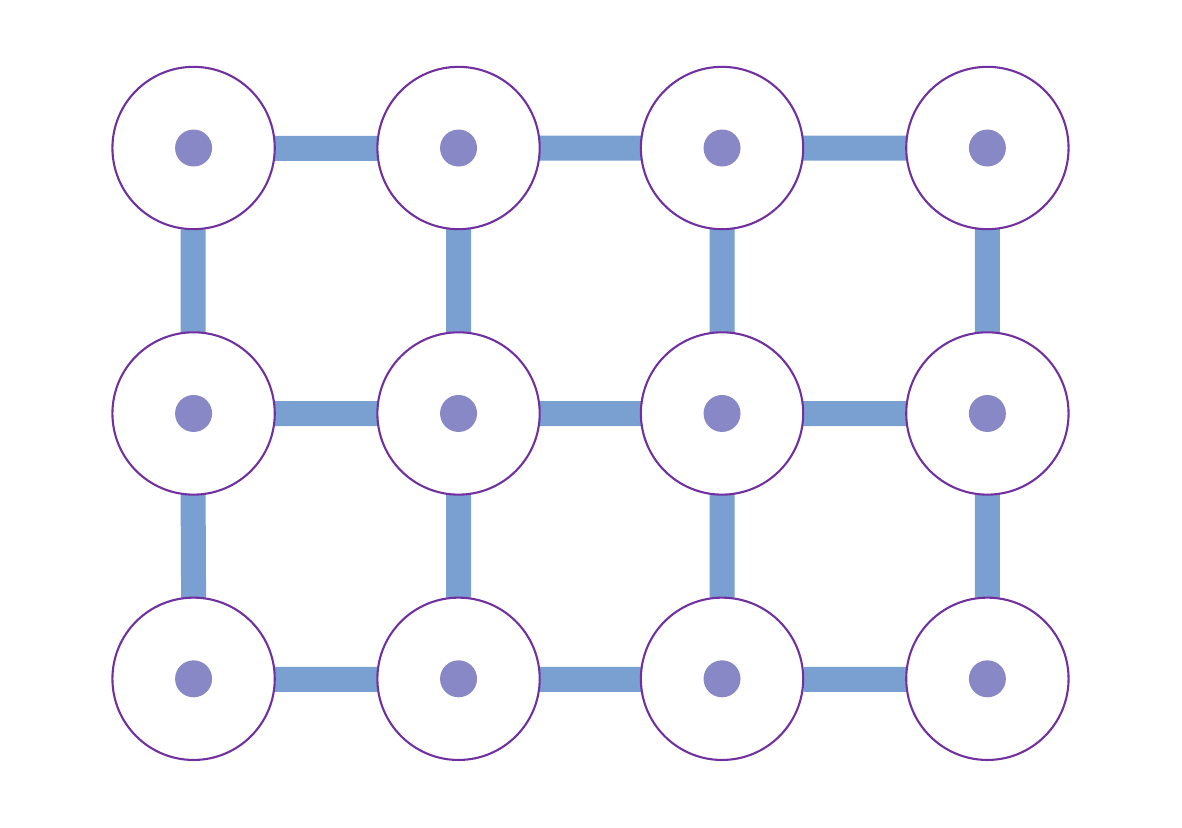}
\caption{A cluster state of $n = 3$ rows and $m = 4$ columns. Each circle represents a qubit in the $\ket{+}$ state and each edge corresponds to the application of a $\CZ$ operator between adjacent qubits.}
\label{fig:cluster-state}
\end{figure}

Bob can generate these states by using coherent pulses to repeatedly excite quantum emitters with a spin degree of freedom\footnote{Any other qubit degree of freedom would work.} and spin-selective transitions which can emit spin-entangled photonic qubits.

More precisely, we denote $\ket{\downarrow}$ and $\ket{\uparrow}$ the computational basis states of the spin qubit. After excitation, we assume that the quantum emitter will spontaneously emit a photon whose polarisation is correlated to the spin state as follows: $\ket{\uparrow} \xrightarrow[]{E_{\rm qe}} \ket{\uparrow, R}$ and $\ket{\downarrow} \xrightarrow[]{E_{\rm qe}}  \ket{\downarrow, L}$. Overall, an excitation with a weak coherent pulse followed by radiative decay yields the following emission operator for our quantum emitter:
\begin{equation}
    E_{\rm qe} = \ket{\downarrow, L}\bra{\downarrow} + \ket{\uparrow, R} \bra{\uparrow},
    \label{eq:quantum_emitter}
\end{equation}
where $\ket{L}, \ket{R}$ denote respectively left and right polarised photon states.\footnote{We associate states $\ket{\downarrow}, \ket{L}$ to $\ket{0}$, and $\ket{\uparrow}, \ket{R}$ to $\ket{1}$.}
A simple example of such transitions is displayed in Fig.~\ref{fig_transition} and corresponds to a singly-charged single quantum dot.

\begin{figure}[ht]
    \centering
         \includegraphics[width=5cm]{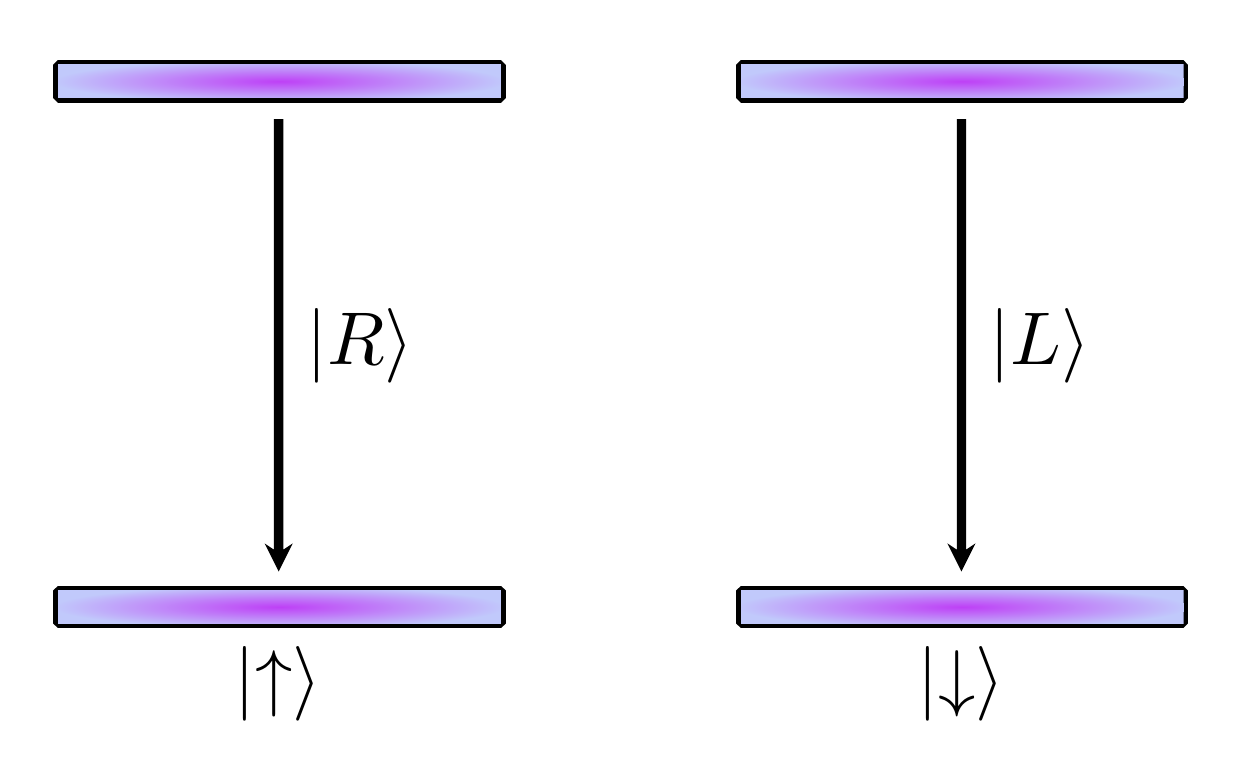} 
    \caption{Energy level structure of a quantum emitter suitable for our BDQC protocol with semi-classical light communications, such as a negatively-charged quantum dot.
    }
    \label{fig_transition}
\end{figure}

The operation $E_{\rm qe}$ can equivalently be expressed as applying a $\CNOT$ between the quantum emitter and a photon in the $\ket{0}_{ph}$ state as in the following circuit:

\[
\Qcircuit @C=2em @R=1.5em {
  \lstick{\ket{\psi}_{qe}} & \ctrl{1} & \qw & \push{\rule{0em}{0.1em}} \\
  \lstick{\ket{0}_{ph}}    & \targ    & \qw & \rstick{\raisebox{2.7em}{~ $E_{\rm qe}\ket{\psi}_{qe}$} }
\gategroup{1}{2}{2}{4}{0.5em}{\}}
}
\]

Starting from a spin in state $\ket{+}_{qe} := \sfrac{1}{\sqrt{2}}(\ket{\downarrow}+\ket{\uparrow})$, applying the emission operator generates a Bell state including one photon and the spin qubit:
\begin{equation}
\ket{\Psi} = \frac{1}{\sqrt{2}}(\ket{\downarrow}\ket{L} + \ket{\uparrow}\ket{R}).
\end{equation}
This Bell state can be seen as a redundantly encoded $\ket{+}$ state. Applying a Hadamard gate on the spin then generates a two-qubit linear cluster state where one vertex contains the spin qubit and the other is comprised of polarisation qubit:
\begin{equation}
\CZ_{qe, ph}\ket{+}_{qe}\ket{+}_{ph}.
\end{equation}
This process can be represented as in the following Figure \ref{fig:basic-em}:

\begin{figure}[ht]
\centering
\includegraphics[width=2\columnwidth/3]{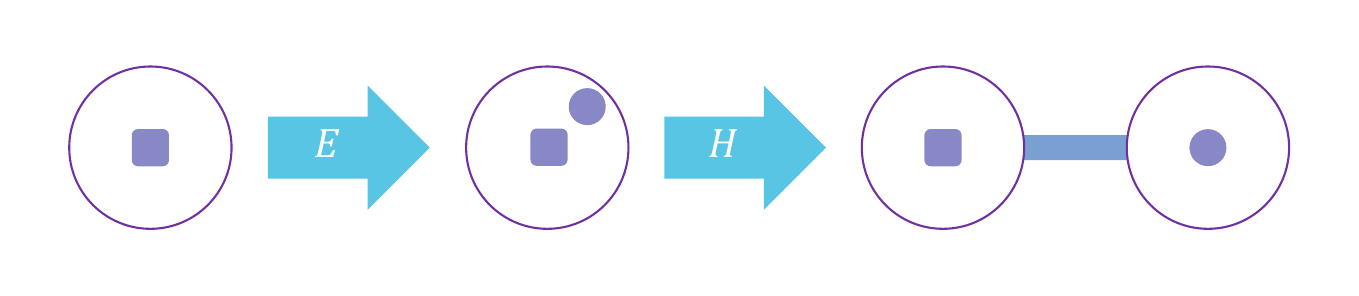}
\caption{Basic linear cluster generation process. Starting from an appropriate quantum emitter in the $\ket{+}_{qe}$ state (purple square), a Bell state $\ket{\Psi}$ is created via emission of a photon (purple dot) and then transformed into a cluster state. The middle circle containing two elements represents the Bell state as a redundantly-encoded $\ket{+}$ state.}
\label{fig:basic-em}
\end{figure}

Intuitively, we can see that applying the emission operator adds a photonic qubit inside the redundantly-encoded vertex containing the quantum emitter, while applying a Hadamard gate to the quantum emitter during this state creation process ``pushes'' the quantum emitter to a new vertex which is linked only to the one it previously occupied. This is exemplified in Figure~\ref{fig:consecutive-em}.

\begin{figure}[ht]
\centering
\includegraphics[width=\columnwidth]{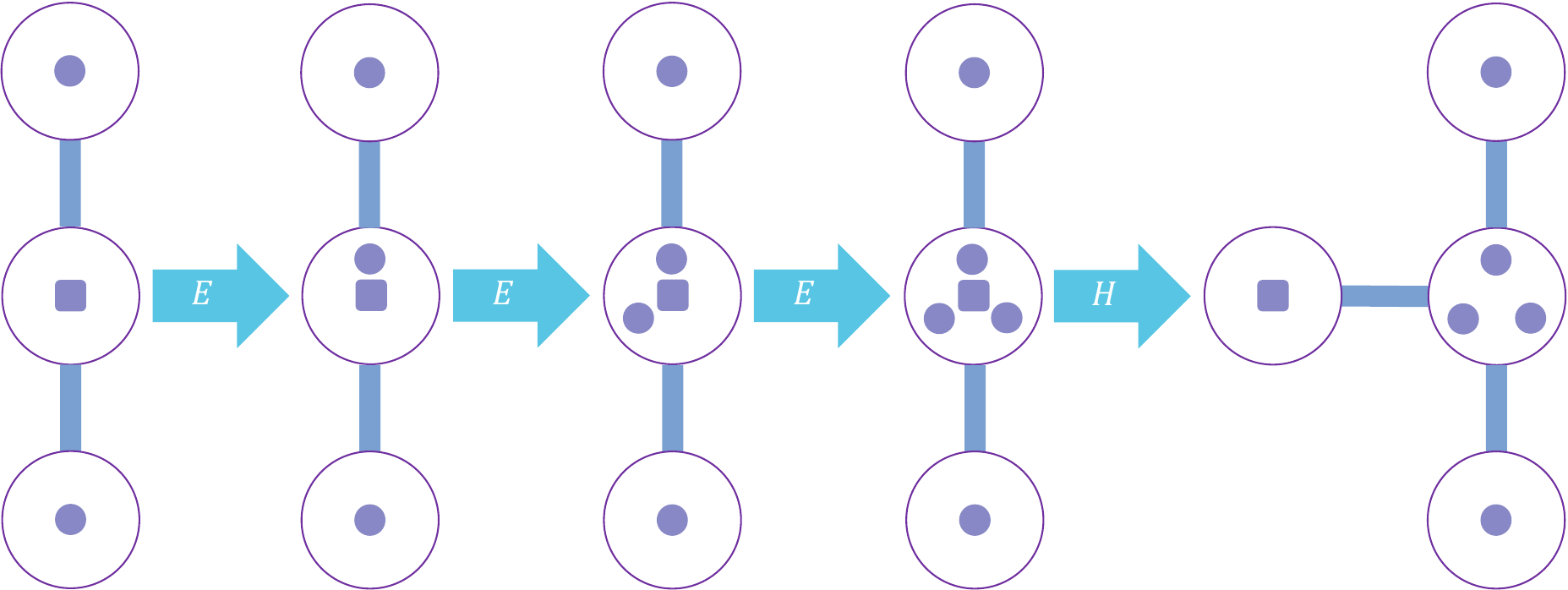}
\caption{Example showing the result of applying the emission operator three times followed by a Hadamard gate. The graph is initially comprised of three vertices, the middle one containing the quantum emitter and the other two a polarisation encoded qubit each. The three emissions generate a 4-qubit redundantly-encoded vertex (i.e. a 4-qubit GHZ state linked to the rest of the graph), while the Hadamard gate applied to the quantum emitter takes it out of the vertex where it was previously and places it in a new vertex.}
\label{fig:consecutive-em}
\end{figure}

To generate a cluster state with $n$ rows and $m$ columns, we use $n$ quantum emitters placed in a line which can each quantumly interact with its neighbours via a $\CZ$ operation. We first apply a $\CZ$ between each quantum emitter and its neighbours, then make each quantum emitter emit one photon and apply a Hadamard on the spin of the emitter. This process is repeated $m$ times, after which the quantum emitters are measured in the $\Z$ basis. Up to single-qubit $\Z$ corrections, this disentangles them from the generated cluster state.\footnote{The $\Z$ basis measurement can be done by emitting an auxiliary photon using $E_{\rm qe}$ and measuring it in the $\Z$ basis.}
This process is presented in Figure \ref{fig:cluster-gen}.

\begin{figure}[htp]
\centering
\subfloat[First step in the cluster state emission. Starting from $n = 3$ quantum emitters, we apply a $\CZ$ gate between each pair of neighbours, followed by an emission operator and Hadamard on all emitters.]{
\label{fig:cluster-gen-a}
\includegraphics[height=5.5cm]{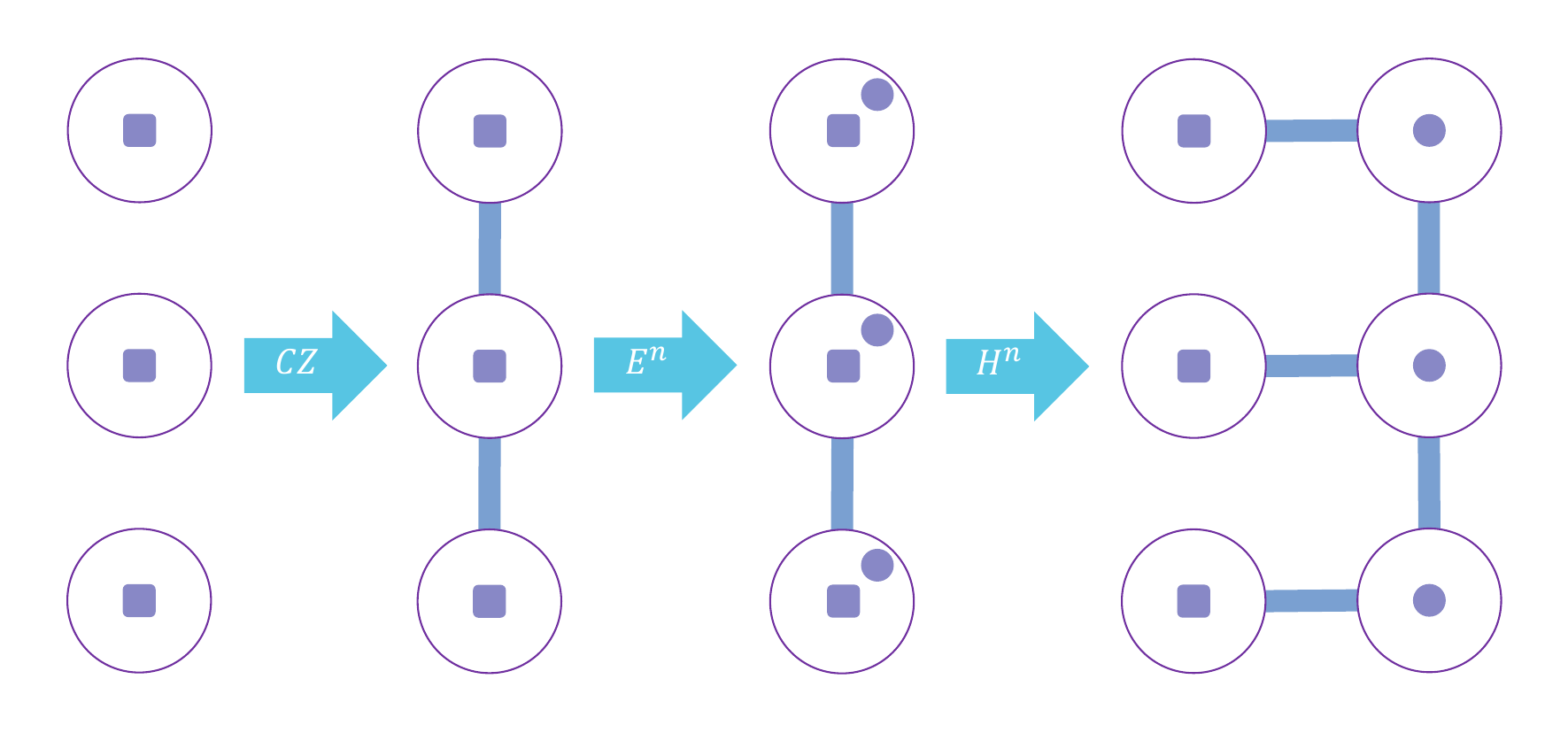}
}\\
\subfloat[Result after repeating the process in Figure \ref{fig:cluster-gen-a} $m = 4$ times. The quantum emitters are then measured in the $\Z$ basis. This yields the state from Figure \ref{fig:cluster-state} encoded in the polarisation of the emitted photons.]{
\label{fig:cluster-gen-b}
\includegraphics[height=5.5cm]{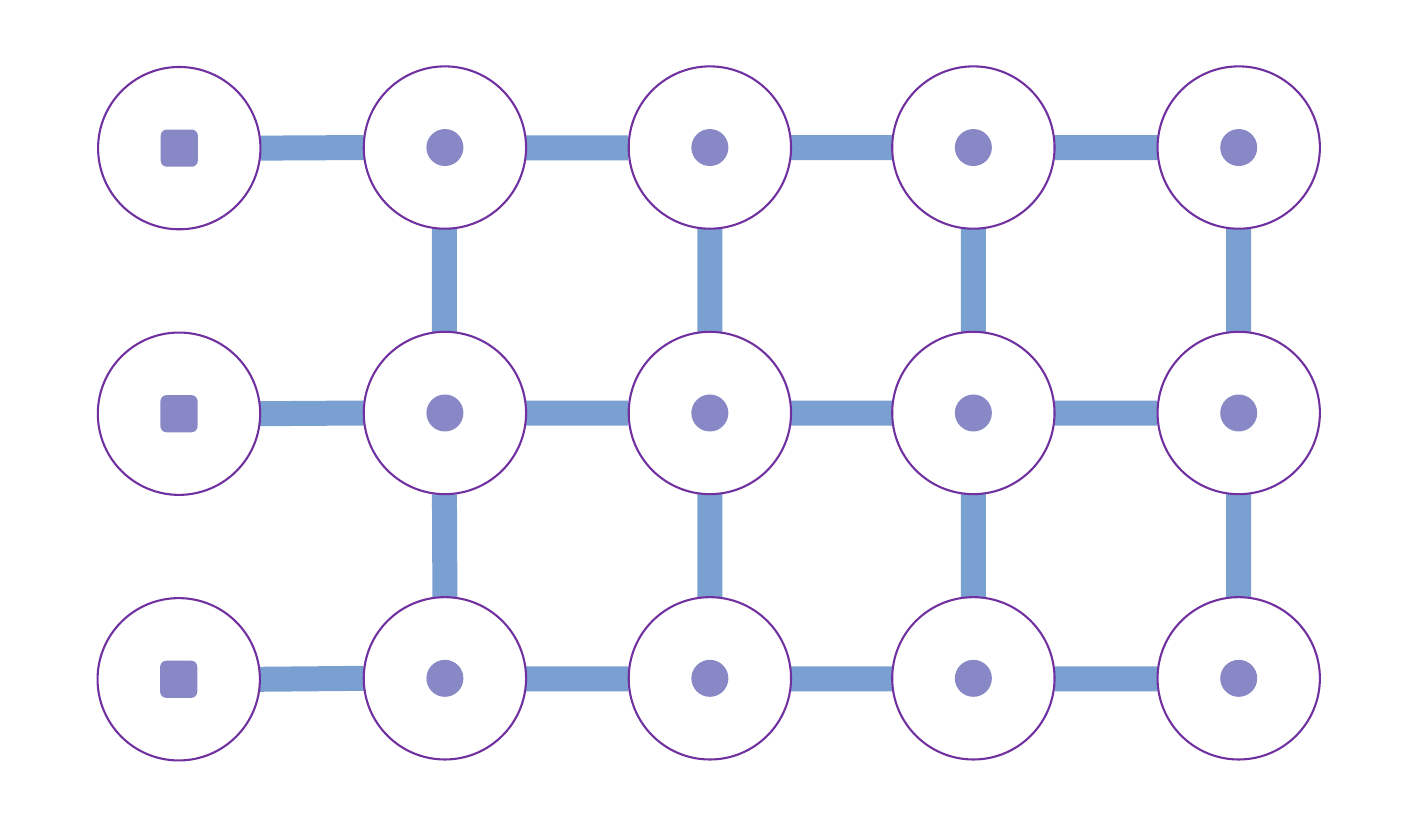}
}
\caption{Generation of a cluster state of $n = 3$ rows and $m = 4$ columns.}
\label{fig:cluster-gen}
\end{figure}

Alternatively, arbitrary graph states can be constructed from linear clusters using various schemes, for example using non-interacting spin quantum emitters and linear optics as in~\cite{Hilaire2022}. Each quantum emitter generates a redundantly-encoded linear cluster state. This is done by emitting $N$ photons before applying a Hadamard gate. After $n$ repetitions, this results in the following state $\bigotimes \CZ_{i, i+1} \ket{\GHZ_N}^{\otimes n}$, where each $\CZ$ gate acts on one qubit from two neighbouring $N$ qubit GZH states defined as $\ket{\GHZ_N} := \sfrac{1}{\sqrt{2}}(\ket{L}^{\otimes N} + \ket{R}^{\otimes N})$. Such redundantly-encoded linear cluster states can be represented as in Figure \ref{fig:red-lin-cluster}.

\begin{figure}[htp]
\centering
\includegraphics[width=2\columnwidth/3]{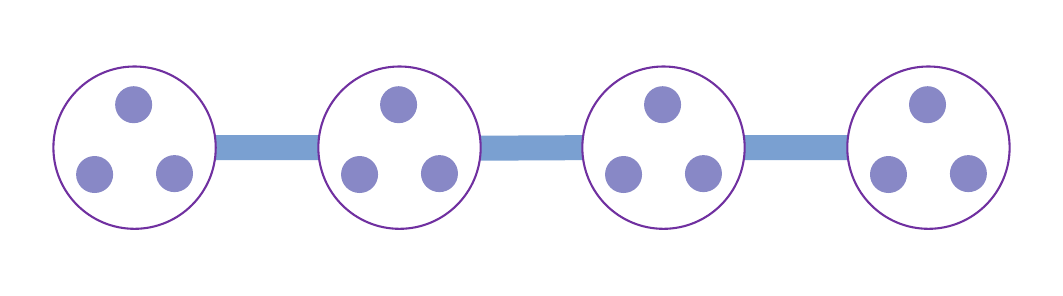}
\caption{A $3$-qubit redundantly-encoded linear cluster state of $m = 4$ columns. Each circle with three dots represents a three-qubit GHZ state.}
\label{fig:red-lin-cluster}
\end{figure}

Two such states can be joined together using fusion operations: to create an edge between two vertices, a pair of photons -- one from each of the vertices -- are measured using a fusion gate (similar to a Bell measurement)~\cite{Browne2005}. A success joins the two states, while a failure means simply losing these photons. This process can be repeated so long as there are photons left in the GHZ states.
The success probability is quite sensitive to photon losses after emission -- additional photons are created and must be detected to create some edges of the graph. However, in principle, this scheme can be made near-deterministic in the absence of loss.

Our secure protocol presented later in the paper can also be straightforwardly adapted to the graph state generation techniques based on deterministic spin-photon interactions from~\cite{Pichler2017, Zhan2020, Wan2021, Shi2021}. Note that we have used a polarisation encoding based on a four-level system as depicted in Fig.~\ref{fig_transition}, but we can also easily generalise the generation schemes to other types of photon encoding and/or quantum emitter level structure~\cite{Lee2019, Tiurev2021, Tiurev2022, Appel2022}. 

\subsection{Delegated quantum computing protocols}
\label{subsec:ubqc}

\subsubsection{Measurement-based quantum computing}

We start by describing how computations are performed in the MBQC model for classical inputs and outputs. Alice's computation is defined by a so-called \emph{measurement pattern} which consists of a graph $G = (V, E, I, O)$, with vertex and edge sets $V$, $E$, input vertex set $I$ and output vertex set $O$, a partial order $\preceq$ over $V$, a set of angles $\{\phi_v\}_{v \in V}$ and a flow function $f$ from non-output vertices $V \setminus O$ to non-input vertices $V \setminus I$.

The computation starts by producing on Bob's machine the graph state $\ket{G}$ associated to the graph $G$: for each vertex in $V$ one qubit is prepared by Alice in the state $\ket{+}$ and sent to Bob, who then applies a $\CZ_{v, w}$ operation between each pair of qubits associated to vertices $v, w$ joined by an edge in $E$.

Alice can then perform a quantum computation on the graph state generated by Bob by instructing him to measure vertices in the order $\preceq$. These measurements are performed in the $\X - \Y$ plane of the Bloch sphere, i.e.~measurements in the $\ket{\pm_{\phi}} = \sfrac{1}{\sqrt{2}}(\ket0 \pm e^{i \phi} \ket1)$ bases. The measurements are specified for each vertex $v$ by the angle $\phi_v$.

If a measurement outcome on a non-output vertex is $1$, some future measurement angles must be adapted in order to perform the same computation as if the measurement outcome had been $0$. The flow function $f$ prescribes how the base angles $\{\phi_v\}_{v \in V}$ change due to this effect, see~\cite{Browne2007} for more details. We will note $\phi'_v$ the updated angle.

In the case of classical outputs, the outcome of the computation is given by the measurement results on vertices from set $O$ after corrections given by the flow are applied.

This scheme is universal when using specific graphs, either brickwork graph states~\cite{Broadbent2010} or cluster states~\cite{Mantri2017}, and measurement angles $\phi_v$ chosen from the set $\Phi = \{ j \pi / 4 \mid j = 0, 1,\ldots, 7 \}$.\footnote{MBQC is exactly universal if we allow arbitrary angles from $[0, 2\pi)$, and it is approximately universal for $2^k$ angles of the form $2 j \pi / 2^k$ so long as $k \geq 3$, see~\cite{Broadbent2010}.}

Following this procedure allows Alice to delegate any quantum computation of her choice, at the cost of revealing it entirely to Bob.

\subsubsection{Blind delegated protocol}

The goal of the UBQC protocol is to allow Alice to perform the same task but without Bob learning anything about her desired computation. In the original UBQC protocol~\cite{Broadbent2010}, Alice also sends qubits to Bob, but their initial state is chosen among $\ket{+_\theta}= \sfrac{1}{\sqrt{2}}(\ket0 + e^{i \theta} \ket1)$ with $\theta \in \Phi$. We can rewrite this state as $\ket{+_\theta} = \RZ(\theta) \ket{+}$, with:
\begin{equation}
\RZ(\theta) = \exp\left(-\frac{i \theta}{2} \Z\right) = \cos(\theta/2) \Id - i \sin(\theta/2) \Z
\end{equation}

Since $\RZ(\theta)$ commutes with $\CZ$ applied by Bob after receiving Alice's qubits, the resulting state is a graph state up to local rotations $\RZ(\theta)$ on each qubit:

\begin{align}
    \ket{G(\vec \theta)} & = \left(\prod_{(v,w) \in E} \CZ_{v,w} \right) \RZ(\vec \theta) \otimes_{v \in V} \ket{+}_v = \RZ(\vec \theta) \ket{G}, \label{eq_theta_G} \\
    \intertext{with}
    \RZ(\vec \theta)& = \prod_{v \in V}\RZ^v(\theta_v)
\end{align}
where $\vec \theta$ is a vector corresponding to all the angles $\theta_v$ applied on each vertex $v \in V$, and the $v$ index in $\RZ^v$ denotes on which qubit the rotation is applied. We will call any graph state of the form $\ket{G(\vec \theta)}$ a \textit{blind graph state}.

To perform the same computation as before, instead of measuring the qubit at vertex $v$ using the measurement angle $\phi'_v$, Alice will instead request that Bob use the angle $\delta_v = \phi'_v + \theta_v + r_v \pi$. Consequently, since the measurement with angle $\delta_v$ is performed by applying $\RZ(-\delta_v)$ and measuring in the $\ket{\pm}$ basis, the $\theta_v$ from the measurement angle will cancel out the $\theta_v$ applied during the state generation and apply Alice's desired computation. It will at the same time perfectly hide the value of $\phi'_v$ in $\delta_v$. The value of $r_v$ is chosen uniformly at random from $\{0, 1\}$ to randomise the outcome of the measurement. This is required since Bob can otherwise learn the outcome of the computation. So long as $\vec \theta$ and $\vec r$ are hidden, the quantum computation is blind and leaks no information as proven in Ref.~\cite{Broadbent2010}.

For completeness sake, we describe in Appendix~\ref{app:dqc-protocols}, Protocol~\ref{prot:UBQC} the full UBQC protocol for classical inputs and outputs and provide the formal security statement in Appendix \ref{app:ac}, Theorem \ref{thm:sec-ubqc}.

\subsubsection{Secure delegated protocol}
\label{subsubsec:sdqc}

Secure Delegated Quantum Computations (SDQC) protocols not only ensure that Alice's information remains private, but also provide verifiability. This property guarantees that Alice either receives the correct output or aborts the computation due to Bob deviating from his prescribed operations.

Various techniques have been developed for inserting tests alongside Alice's computation of interest to check that Bob's behaviour is honest~\cite{Fitzsimons2017a,Kashefi2017,Leichtle2021}, but most require the preparation of both rotated $\ket{+_\theta}$ states and states in the computational basis. These are used to isolate qubits in the graph so that Bob's measurement of these qubits returns a deterministic result if performed correctly. Such isolated qubits can then be used to test Bob throughout the execution of the computation.

We present here the protocol from \cite{Kapourniotis2023} for classical input/output $\mathsf{BQP}$ computations\footnote{$\mathsf{BQP}$ computations have the property that the probability distribution of classical outcomes from correct executions will contain an outcome that has a probability $1/2 + c$ of appearing, with $c > 1/p(n)$ for $p$ a polynomial in the size of the input $n$.} which uses only $\ket{+_\theta}$ states. Alice and Bob repeat the same UBQC execution $N$ times, using a random fraction of the executions to test Bob's honesty via deterministic computations whose outcomes are known to Alice. If less than a given fraction of these tests pass, Alice aborts. Otherwise she computes the majority outcome of the UBQC runs which performed her desired computation.

The tests use the fact that the graph state $\ket G$ is the common eigenstate of all Pauli operations of the form $\mathsf{S}_v = \X_v \prod_{(v,w) \in E} \Z_w$ -- called the graph's stabilisers. A test is then defined by a product of these stabilisers $\mathsf{S}_{\vec t} = \prod_{v \in V} \mathsf{S}_v^{t_v}$ for a bit string $\vec t \in \{0,1\}^{\abs{V}}$. To perform the test associated to stabiliser $\mathsf{S}_{\vec t}$, Alice instructs Bob to measure all vertices in the Pauli bases given by this stabiliser. She then computes the parity of measurement outcomes and the test passes if it is $0$.

In UBQC, Bob only performs measurements in the $\X - \Y$ plane, therefore Ref.~\cite{Kapourniotis2023} requires that Alice only uses tests $\mathsf{S}_{\vec t}$ such that $\mathsf{S}_{\vec t}$ does not act as Pauli $\Z$ on any vertex $v$. Then the state preparation and measurements for tests and computations are indistinguishable since both can be hidden via a UBQC overlay. For various types of graphs including the brickwork and cluster graphs, Ref.~\cite{Kapourniotis2023} then defines sets of such tests from which Alice draws at random when she must perform a test run during the protocol's execution.

Ref.~\cite{Kapourniotis2023} show that these tests can catch all of Bob's deviations which may corrupt Alice's computation in the case of classical outputs. Then, if $N$ is the total number of UBQC executions, the probability that Alice accepts an incorrect outcome at the end of the SDQC protocol decreases exponentially with $N$.

For completeness sake, we describe in Appendix \ref{app:dqc-protocols}, Protocol~\ref{prot:SDQC} the full SDQC protocol for classical inputs and outputs, and provide the formal security statement in Appendix \ref{app:ac}, Theorem \ref{thm:sec-sdqc}.

\section{Remote state preparation with semi-classical clients}
\label{sec:rsp-semi-cl}

The security of the delegated quantum computation protocols presented in the previous section rests on the fact that Bob cannot access information about the secret angle $\theta_v$ with certainty. 
When sharing single qubits, this is ensured by the no-cloning theorem~\cite{Wootters1982, Dieks1982}. 
However, single-qubit sources are expensive and may be inaccessible to most clients who desire to run secure quantum computations. We describe in this section how the same result can be obtained using limited means on both sides of the protocol.

\subsection{Blind graph state generation using quantum emitters}
\label{subsec:blind-graph-rsp}

Instead of sending single qubits, in our protocol Alice transmits coherent states to hide the phases $\vec \theta$. The main advantage is that a photonic attenuated coherent state source is simply a laser light which is already commercially available and likely to remain much more economical than a source of single qubits.
Such a coherent state $\ket{\alpha}_\theta$ with $\alpha = \abs{\alpha} e^{i\varphi}$ is given by:
\begin{equation}
    \ket{\alpha}_{\theta} = e^{- \abs{\alpha}^2/2} \sum_{k=0}^\infty \frac{\alpha^k}{\sqrt{k!}} \ket{k}_\theta,
\end{equation}
for $\ket{k}_\theta := \left(a_\theta^\dagger\right)^k\ket{\upsilon}$, where $a_\theta^\dagger := \cos(\theta)a_R^\dagger + \sin(\theta)a_L^\dagger$ is the rotated polarisation creation operator and $\ket{\upsilon}$ is the vacuum state.
Using a phase-randomised laser, Alice instead sends the state:
\begin{align}
    \rho_{\alpha,\theta} &= \frac{1}{2\pi} \int_{\varphi=0}^{2\pi} \op{\alpha}_\theta \mathrm{d}\varphi\\
    &= \frac{1}{2\pi} \int_{\varphi=0}^{2\pi} \op{\abs{\alpha}e^{i \varphi}}_\theta \mathrm{d}\varphi \\
    &= e^{- \abs{\alpha}^2} \sum_{k=0}^\infty \frac{\abs{\alpha}^{2k}}{k!} \op{k}_\theta.
\end{align}
Randomising the laser's phase is responsible for the sum over all possible phase values. While the coherent state $\ket{\alpha}_\theta$ is a coherent superposition, this randomisation yields a mixed state such that the probability of having $k$ copies of the same state $\ket{1}_\theta = \ket{+_\theta}$ is $e^{- \abs{\alpha}^2}\abs{\alpha}^{2k}/k!$.\footnote{This will be useful in proving the security of our protocols since it will allow us to bound the probability that more than a single copy of the state has been received. This would not be possible with a coherent superposition such as the one before the phase randomisation.} We assume that $\abs{\alpha}^2$ is the effective intensity of the laser pulse as it enters Bob's setup, having factored in the losses induced by the communication channel. Any loss discussed later will correspond to those that occur after this pulse has arrived.

On the other hand, Bob must convert these coherent states into single qubits while preserving the polarisation in order to use them in the UBQC protocol. This can be done using the scheme presented in Section~\ref{subsec:graph-gen} with a slight modification of the emission operator in Eq.~\eqref{eq:quantum_emitter}.
If Alice's light source sends the state $\rho_{\alpha,\theta}$ at an intensity $\abs{\alpha}^2$ such that it is able to drive the optical transition of the quantum emitter from Figure~\ref{fig_transition}, the production of spin-entangled photon states is done via the following global operation:
\begin{equation}
    E_{\rm qe}(\theta) = \ket{\downarrow, L}\bra{\downarrow} + e^{i \theta} \ket{\uparrow, R} \bra{\uparrow} = \RZ(\theta) E_{\rm qe}.
\end{equation}

The polarisation of the input coherent laser therefore imprints a phase on the photon after emission, corresponding to a $\RZ(\theta)$ rotation. 

If the operation $E_{\rm qe}(\theta)$ is used instead of simply $E_{\rm qe}$ in the graph generation process described in Section~\ref{subsec:graph-gen}, then the resulting state is the rotated cluster state $\ket{G(\vec \theta)}$, which is the required resource for both the UBQC and SDQC protocols as shown in Eq.~\eqref{eq_theta_G}.
Intuitively, so long as all the angles $\theta_v$ are hidden from Bob, the UBQC protocol which uses this strategy for generating the blind graph remains secure. We now introduce a new resource in the Abstract Cryptography security framework and then describe a protocol which formalises this idea. For unitary $\mathsf{U}$, we write $\mathsf{U}(\rho)$ to mean $\mathsf{U}\rho\mathsf{U}^\dagger$. Furthermore, when applying $\CNOT$ operations, we always assume that the first qubit mentioned is the control while the second one is the target.

\begin{resource}[ht]
\caption{Blind Graph State Extender}
\label{res:graph-extend}

\begin{algorithmic}[0]

\STATE \textbf{Inputs:} Alice inputs an angle $\theta \in \Phi$. Bob inputs a single-qubit state $\rho_B$.

\STATE \textbf{Computation by the resource:} The resource samples a bit $b \in \bin$ uniformly at random and sends to Bob the two-qubit state $\RZ((-1)^b\theta)\CNOT\left(\rho_B\otimes\op{0}\right)$, where $\RZ$ is applied to the second qubit.

\end{algorithmic}
\end{resource}

This resource captures the perfect setting where the rotated emission operator $E_{\rm qe}(\theta)$ is applied to Bob's quantum emitter without leaking any information about the state. The additional value $b$ gives more flexibility in how this resource can be implemented. We then show that this resource can be used to securely produce a blind graph state on Bob's device via the following Protocol~\ref{prot:graph-rsp}. We do not specify the entanglement operations performed by Bob so that it may use either strategies outlined in section \ref{subsec:graph-gen}. To each of Bob's quantum emitters $q$, we associate the ordered list of vertices $V_q \subset V$ which Bob generates using this quantum emitter. We assume that for two consecutive vertices $v, w \in V_q$ we have $(v, w) \in E$.

\begin{protocol}[htp]
\caption{Blind graph state preparation from blind extensions}
\label{prot:graph-rsp}
\begin{algorithmic} [0]

\STATE \textbf{Public information:} Graph $G = (V, E)$ and an ordering $<$ over vertices decided by Bob.

\STATE \textbf{Inputs:} Alice inputs a set of angles $\vec \theta \in \Phi^{\abs{V}}$.

\STATE \textbf{Protocol:}

\begin{enumerate}
\item Bob initialises the spin qubit of all his quantum emitters in the state $\ket{+}_{qe}$.
\item For each vertex $v$ in the graph, in the order $<$:
\begin{enumerate}
\item Alice and Bob perform a call to the Blind Graph State Extender Resource~\ref{res:graph-extend}:
\begin{itemize}
\item Alice inputs the angle $\theta_v$. Bob inputs the spin qubit of quantum emitter $q$ such that $v \in V_q$.
\item Bob receives as output two qubits along with a bit $b_v$.
\end{itemize}
\item Bob applies $\X^{b_v}$ to both qubits it received and $\Z^{b_v}$ to any one qubit at the vertex $v' \in V_q$ preceding $v$, if it exists.
\item Bob produces a number $n_v$ of his choice of additional qubits using the emission operator $E_{\rm qe}$.
\item Bob performs all entanglement operations between vertex $v$ and the vertices $w < v$ such that $(v, w) \in E$. These can be performed either via:
\begin{itemize}
\item $\CZ$ gates between quantum emitters $q$ and $r$ where $v \in V_q$ and $w \in V_r$, if $r$ is still situated at vertex $w$;
\item fusion operations between the additional qubits generated by Bob at vertices $v$ and $w$.
\end{itemize}
\item Bob applies a Hadamard gate to the spin qubit of quantum emitter $q$.
\end{enumerate}
\item Bob decorrelates his quantum emitters from the polarisation-encoded graph state by performing:
\begin{enumerate}
\item a measurement of each spin qubit $q$ in the computational basis with outcome $c_q$;
\item a correction $\Z^{c_q}$ on a qubit of vertex $v_q$ where $v_q$ is the last element of $V_q$.
\end{enumerate}
\item If there are leftover additional qubits at any vertex $v$, Bob measures them in the $\ket{\pm}$ basis. Let $d_v$ be the sum of measurement outcomes, Bob applies $\Z^{d_v}$ to the remaining qubit at vertex $v$.
\end{enumerate}

\end{algorithmic}
\end{protocol}

\paragraph{Protocol security.}
The blind graph state generation protocol (Protocol~\ref{prot:graph-rsp}) perfectly constructs the Blind Graph RSP Resource~\ref{res:rsp}, meaning that the security error is exactly $0$. This is captured by Theorem~\ref{thm:sec-graph-rsp}, whose proof can be found in Appendix \ref{app:sec-proof-graph}.

\begin{theorem}[AC security of Protocol~\ref{prot:graph-rsp}]
\label{thm:sec-graph-rsp}
\newcounter{count:sec-graph}
\setcounterref{count:sec-graph}{thm:sec-graph-rsp}
Protocol~\ref{prot:graph-rsp} perfectly constructs in the Abstract Cryptography framework the Blind Graph RSP Resource~\ref{res:rsp} from $\abs{V}$ instances of the Blind Graph State Extender Resource~\ref{res:graph-extend}.
\end{theorem}

\subsection{Security amplification via rotation aggregation}
\label{subsec:amplification}

It would seem like simply exciting a quantum emitter via a pulse in the state $\rho_{\alpha,\theta}$ would be enough to construct the Blind Graph State Extender Resource~\ref{res:graph-extend} since it is equivalent to applying the emission operator $E_{\rm qe}(\theta)$. However this is not the case since a laser pulse leaks a non-negligible amount of information.

Due to the phase randomisation step, Bob has access to a probabilistic mixture of states $\op{k}_\theta$. These states consist of $k$ photons which each contain the same information about $\theta$. If he performs a polarisation-preserving photon number resolving QND measurement, he can learn the number $k$ of incoming photons for each vertex of the blind graph state. Bob can then choose to attack the vertex where he received the most photons.

If we assume that Bob is capable of manipulating each of these photons in isolation, the amount of knowledge that he can extract from the execution of the protocol directly increases with the intensity of the coherent state $\abs{\alpha}^2$ via the probability of obtaining multiple photons. On the other hand, a relatively high value for $\abs{\alpha}^2$ is desirable in order to excite the quantum emitter generating the blind graph state.

In order to recover the ideal scenario described by the Blind Graph Extender Resource~\ref{res:graph-extend}, we make use of a security amplification gadget which exponentially suppresses the information leakage using only a linear number of pulses. 

This protocol makes use of the GHZ-generation capability of the quantum emitter paired with the rotated emission operator $E_{\rm qe}(\theta)$. 
Alice sends to Bob a certain number of phase-randomised rotated weak coherent pulses $\rho_{\alpha, \theta_i}$ for a fixed intensity $\abs{\alpha}^2$ and randomly chosen values of $\theta_i$. Bob uses these pulses to drive the transition of his quantum emitter, effectively applying the rotated emission operator $E_{\rm qe}(\theta_i)$. These consecutive pulses will generate a rotated GHZ state whose angle is the sum of all angles used by Alice for laser pulses containing at least one photon. Bob emits one last photon using $E_{\rm qe}$ and then measures all qubits generated by Alice's pulses in the $\X$ basis. He can thus detect which laser pulses contained no photons. If too few photons have been detected, Alice aborts the protocol. This threshold, set by the security proof, prevents Bob from discarding too many pulses from which he would not get enough information. If Alice has not aborted, she communicates to Bob a correction which depends on her desired angle and the parity of the measurement outcomes. After applying this correction, Bob's device contains the random rotated state chosen by Alice. This process is presented in Figure~\ref{fig:protocol}.

\begin{figure}[ht]
\centering
\includegraphics[width=\columnwidth]{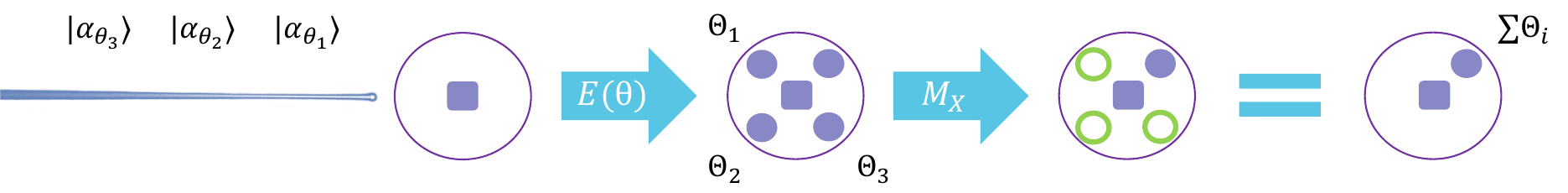}
\caption{Example of privacy amplification with three weak coherent pulses if there are no losses after emission. The laser with polarisation control on the left is in the hands of Alice while the quantum emitter (purple square) is handled by Bob. Alice sends three pulses with randomly rotated polarisations. Bob uses these pulses to drive the quantum emitter and then emits one more photon (purple dots). The result is a five-qubit rotated GHZ state after the first arrow. Bob then measures (green circles) all polarisation qubits in the GHZ state but one in the $\ket{\pm}$ basis. After correction, the result is a rotated Bell state containing the spin qubit and one polarisation qubit. This is the same state -- up to rotation -- as in the middle step of the basic linear cluster generation process from Figure \ref{fig:basic-em}.}
\label{fig:protocol}
\end{figure}

We describe formally this process in Protocol~\ref{prot:ghz-gadget}.

\begin{protocol}[ht]
\caption{GHZ privacy amplification for rotated states from weak coherent pulses}
\label{prot:ghz-gadget}
\begin{algorithmic} [0]

\STATE \textbf{Public information:} Laser pulse intensity $\abs{\alpha}^2$, number $n$ of pulses sent per GHZ gadget, threshold $t$.

\STATE \textbf{Inputs:} Alice inputs an angle $\theta \in \Phi$. Bob inputs a single-qubit spin state $\rho_{qe}$.

\STATE \textbf{Protocol:}

\begin{enumerate}
\item Alice samples values $(\theta_1, \ldots, \theta_n)$ uniformly at random from $\Phi$.
\item For $1 \leq i \leq n$:
\begin{enumerate}
\item Alice sends the phase-randomised rotated weak coherent pulse $\rho_{\alpha, \theta_i}$ to Bob.
\item Bob uses this pulse to apply the rotated emission operator to his spin qubit $E_{\rm qe}(\theta_i)$.
\end{enumerate}
\item Finalisation:
\begin{enumerate}
\item Bob emits another photon using the emission operator $E_{\rm qe}$. This qubit is indexed $0$.
\item Bob attempts to measures all photonic qubits $i \neq 0$ in the GHZ state in the $\ket{\pm}$ basis. Let $S \subseteq \{1,\ldots,n\}$ be the set of indices for which the measurement succeeded (i.e. a photon has been detected) and let $b = \bigoplus_{i \in S} b_i$ be the parity of the corresponding measurement outcomes. Bob sends $S$ to Alice.
\item Alice aborts if $\abs{S} \leq t$, sending $\Ab$ to Bob and setting it as her output. Otherwise, she samples uniformly at random a bit $m_x \in \bin$ and sends the correction angle $\bar{\theta} = (-1)^{m_x}\theta - \sum_{i \in S}\theta_i$ and bit $m_x$ to Bob.
\item Bob applies the correction $\RZ(\bar{\theta} + b\pi)$ to photonic qubit $0$ and sets it as his output together with bit $m_x$.
\end{enumerate}
\end{enumerate}

\end{algorithmic}
\end{protocol}

\paragraph{On the final correction operation.}
The final correction might seem like a superfluous step, but it is necessary for proving the security of the protocol. In particular, the simulator built as part of the security proof requires this flexibility to adapt its behaviour to the set $S$ returned by the distinguisher. However, this does not mean that this correction must necessarily be implemented as an additional quantum operation by Bob. We use this protocol within Protocol~\ref{prot:graph-rsp} to build a resource state for the UBQC protocol, which already requires Alice to send measurement angles to Bob. Alice can merge the correction with these measurement angles so that both are applied simultaneously as as single operation. Therefore, this correction does not imply an additional classical communication step, nor an additional operation by Bob, if Protocol~\ref{prot:ghz-gadget} is used together with UBQC.

\paragraph{Protocol security.}
The privacy amplification protocol presented above constructs the Blind Graph State Extender Resource~\ref{res:graph-extend}, with a security error which decreases exponentially in the number of pulses $n$ sent by Alice. This analysis holds in the lossless regime, assuming that each laser pulse has a probability $\eta_1$ of applying the rotated emission operator to the spin qubit. The next section gives arguments for this assumption and discusses more realistic settings when taking losses into account. The security of our protocol is given by Theorem~\ref{thm:sec-gadget}, whose proof can be found in Appendix \ref{app:sec-proof}.

\begin{theorem}[AC security of Protocol~\ref{prot:ghz-gadget}]
\label{thm:sec-gadget}
\newcounter{count:sec}
\setcounterref{count:sec}{thm:sec-gadget}
Let $\eta_1$ be the probability that the quantum emitter generates a photon after receiving a laser pulse of intensity $\abs{\alpha}^2$ and let $p_{\alpha, 2} = 1 - e^{-\abs{\alpha}^2} - \abs{\alpha}^2 e^{-\abs{\alpha}^2}$ be the probability that the laser pulse contains two or more photons. Let $(\alpha, n)$ be the public parameters used in Protocol~\ref{prot:ghz-gadget}, and let $t = \frac{\eta_1 + p_{\alpha, 2}}{2}n$. Then Protocol~\ref{prot:ghz-gadget} $\exp(-\nu_{\alpha}n)$-constructs in the Abstract Cryptography framework the Blind Graph State Extender Resource~\ref{res:graph-extend}, for $\nu_\alpha = (\eta_1 - p_{\alpha, 2})^2/2$.
\end{theorem}

\subsection{Post-selecting on receiving all photons}
\label{subsec:threshold-vs-post}

It is possible also for Alice to reject the state as soon as a single photon is lost. We describe this alternative formally in Appendix~\ref{app:post-select-prot} as Protocol~\ref{prot:post-select-ghz-gadget}, along with a correctness and security analysis.

The threshold protocol presented above will abort significantly less often than the one which requires Bob to recover all photons. This is because Protocol~\ref{prot:ghz-gadget} tolerates some losses due to a photon not being emitted. On the other hand, the security bound is slightly degraded. Combined together, the threshold protocol's security and successful completion probability both increase with the number of emitted photons, as expressed in Theorem~\ref{thm:sec-gadget}. It would seem like it is therefore the superior protocol since the success probability of the post-selected protocol actually decreases exponentially with the number of photons which need to be collected.

However this analysis holds only if we assume that (i) there are no losses after the photons have been emitted, and (ii) the emission operator correctly describes the result of the interaction between the laser pulse and the quantum emitter.

If a photonic qubit is lost after it has been generated in Protocol~\ref{prot:ghz-gadget}, the state is perfectly mixed from the point of view of both Alice and Bob. Discarding the state as soon as a single photon is not detected by Bob during the measurement prevents this issue and as a consequence the correctness of the output state of Protocol \ref{prot:post-select-ghz-gadget} is unaffected by losses which happen after the qubit generation procedure. This includes the case where the photon is emitted while the interaction between the laser and the quantum emitter is still taking place and is filtered out as described in Section~\ref{subsec:two-lev}.

Both protocols on the other hand may suffer if the laser-spin interaction deviates from the ideal emission operator described here. If for instance a photon is emitted while the laser interaction is ongoing and the laser re-excites the quantum emitter resulting in a second photon being emitted, neither protocols will count this as a loss but the end result will be a perfectly mixed state in both cases. However, if the operator instead applies a slightly different angle for some or all values of $\theta$, this can be mitigated if the overall protocol can tolerate this noise level.

\subsection{Delegated quantum computation protocols with semi-classical client}
\label{subsec:dqc-semi-classical}

We describe in this section how to build delegated quantum computation protocols with a semi-classical client. To this end we compose our bind graph RSP protocol and privacy amplification protocol with the protocols from \cite{Broadbent2010} and \cite{Kapourniotis2023} to obtain respectively BDQC and SDQC protocols for classical inputs and outputs.

\paragraph{Blind delegation from semi-classical client RSP.}
We use the composition theorem from the AC framework to replace the call to the Blind Graph RSP resource in the UBQC protocol from Section~\ref{app:sec-proof} with our Protocols~\ref{prot:graph-rsp} and \ref{prot:ghz-gadget}. The general composition Theorem \ref{thm:compos} allows us to combine Theorems~\ref{thm:sec-ubqc}, \ref{thm:sec-graph-rsp} and \ref{thm:sec-gadget} into the following result. 

\begin{theorem}[BDQC with semi-classical client]
\label{thm:semi-classical-bdqc}

Let $\alpha, \nu_\alpha, n$ be defined as in Theorem \ref{thm:sec-gadget} and let $\abs{V}$ be the number of vertices in the graph supporting Alice's desired MBQC computation. Consider the protocol obtained by replacing the blind graph state generation in the UBQC Protocol~\ref{prot:UBQC} (steps 1 and 2) with Protocol~\ref{prot:graph-rsp} composed with Protocol~\ref{prot:ghz-gadget} with parameters $(\alpha, n)$. This protocol $\abs{V}\cdot\exp(-\nu_{\alpha}n)$-constructs in the Abstract Cryptography framework the BDQC Resource~\ref{res:bqc}.

\end{theorem}

\paragraph{Testing the server's honesty with semi-classical light.}
There are two challenges when trying to use our semi-classical client technique together with protocols for SDQC which require the client to produce states in the computational basis, such as those from~\cite{Fitzsimons2017a,Kashefi2017,Leichtle2021}.

First of all, generating a computational basis state with the energy structure presented in Figure \ref{fig_transition} would project the spin state to the corresponding computational basis state. This would then need to be reinitialised by the Server to continue the entanglement generation. The information of whether or not the spin needs to be reinitialised in a superposition would therefore leak whether the state produced is rotated or in the computational basis.

Then, there is the issue of amplifying the security of both computational basis and rotated states using a single gadget which does not introduce new attack vectors or leak information. This has been attempted in~\cite{Kapourniotis2021} but a new attack on their scheme has been shown in~\cite{Kapourniotis2023}. Thankfully, this second work also fixes the issue by proposing a novel SDQC protocol for classical input and output $\mathsf{BQP}$ computations which makes use of the generalised tests introduced in~\cite{Kapourniotis2022}. This corresponds to the protocol presented in Section \ref{subsubsec:sdqc} and more formally as Protocol \ref{prot:SDQC}.

The composable security of their protocol scales exponentially with the number of repetitions of the UBQC protocol, and the security of our blind graph RSP protocol degrades linearly in the number of repetitions. We again use the general composition Theorem \ref{thm:compos} to combine Theorems~\ref{thm:sec-sdqc} and \ref{thm:sec-gadget} into the following result, yielding an SDQC protocol for classical input and output $\mathsf{BQP}$ computations in which Alice only needs to send weak coherent states.

\begin{theorem}[SDQC with semi-classical client]
\label{thm:semi-classical-sdqc}

Let $\alpha, \nu_\alpha, n$ be defined as in Theorem \ref{thm:sec-gadget}, let $N$ be the number of executions of the UBQC Protocol \ref{prot:UBQC} in Protocol \ref{prot:SDQC} and let $\abs{V}$ be the number of vertices in the graph supporting Alice's desired MBQC computation. Consider the protocol obtained by replacing the blind graph state generation step in each UBQC execution in the SDQC Protocol \ref{prot:SDQC} with Protocol~\ref{prot:graph-rsp} composed with Protocol~\ref{prot:ghz-gadget} with parameters $(\alpha, n)$. This protocol $N\abs{V}\cdot\exp(-\nu_{\alpha}n) + \epsilon_S$-constructs in the Abstract Cryptography framework the SDQC for classical input and output $\mathsf{BQP}$ computations Resource~\ref{res:sqc}, for $\epsilon_S$ exponentially decreasing in $N$ being the security bound of the SDQC Protocol \ref{prot:SDQC}.

\end{theorem}

Finally, as shown in \cite{Kapourniotis2022, Kapourniotis2023}, this SDQC protocol tolerates a constant level of global honest noise, meaning that slight defects will be corrected and Alice is able to recover the correct output without aborting even if her apparatus or Bob's have a low-enough level of noise.

\section{Analysis of concrete schemes for protocol implementation}
\label{sec:loss-analysis}

In the previous section we have shown via Theorem~\ref{thm:sec-gadget} that, so long as the total single-photon efficiency $\eta_1$ is greater than the probability of receiving two or more photons $p_{\alpha, 2}$, we can amplify the security of blind state generation to arbitrary levels by increasing the number of photons in the GHZ gadget. 

The vacuum probability of the incoming coherent pulse imposes an upper bound on $\eta_1$ of $p_{\alpha, 1} = 1 - e^{-|\alpha|^2}$. In reality, if the average number of photons $|\alpha|^2$ is not large enough, the population of the emitter will not be sufficiently inverted to achieve single-photon emission with a high efficiency and the criteria for security amplification may not be satisfied. The exact value of $\eta_1$ can depend on the energy-level structure of the device, the excitation scheme, and the pulse parameters. Thus, an important question is: can any practical single-photon generation schemes and parameter regimes allow for $\eta_1 \geq p_{\alpha, 2}$? 

In this section, we will take a step further towards a practical implementation by discussing constraints on $\eta_1$ imposed by basic models for single-photon emission. In particular, we focus on idealised single-photon generation protocols (those that ensure no more than one photon is collected). Experimental considerations for possible implementation, and comparison with the state of the art, will be discussed more in the following section.

\subsection{Two-level emitter}
\label{subsec:two-lev}

Driving a two-level emitter with a laser is a common approach to produce single photons. By applying a coherent pulse with a temporal width much less than the lifetime of the emitter, the state of the emitter can be excited from the ground state $\ket{g}$ to the excited state $\ket{e}$ with a high fidelity, allowing spontaneous emission to produce a single photon. Such an excitation scheme can then be applied to an emitter with an additional internal degree of freedom, such as electron spin, to generate photonic graph states encoded in polarisation \cite{Lindner2009} as discussed in section \ref{subsec:graph-gen}.

The main challenge with using a two-level emitter is to separate the scattered single photon from the laser pulse to achieve a high degree of single-photon purity, where the probability of collecting more than one photon is very small. This is done by using excitation schemes where the scattered single photon and laser pulse can be distinguished by a degree of freedom such as polarisation, frequency, or spatial mode. Unfortunately, achieving a high single-photon efficiency using weak excitation pulses necessarily requires a high degree of light-matter coupling, which implies that these degrees of freedom should be as close as possible to improve the mode-matching between the laser pulse and the emitter. Thus, using a scheme based on a two-level emitter immediately introduces a trade-off between the single-photon efficiency and fidelity for a fixed number of photons in the laser pulse.

To get an idea of the practical limit when using a two-level emitter, and to see how well it could perform compared to the upper bound, we choose to distinguish the single photon from the laser pulse using the temporal degree of freedom. This approach is only practical when using very weak excitation pulses, since otherwise any small imperfections in the temporal shape of the laser pulse would very quickly overwhelm the single photon signal and reduce the single-photon purity. This filtering could be accomplished using a fast optical switch or by using detectors with a sufficiently fast gate or high temporal resolution.

The text-book model for exciting a two-level emitter treats the laser using a semi-classical approximation, where it is assumed that the state of the pulse remains in a coherent state throughout the entire light-matter interaction so that only the quantum dynamics of the emitter must be considered. However, when dealing with weak states of light, measuring a single photon scattered off the emitter can induce non-negligible quantum fluctuations in the laser pulse that may alter the quantum dynamics of the excitation process. Luckily, if we temporally filter only single photons that were emitted after the pulse has completely finished interacting with the emitter, causality ensures that the semi-classical approximation will provide an accurate analytic solution for the single-photon efficiency $\eta_1$.

A basic model for a driven two-level emitter using the semi-classical approximation is described by the quantum optical master equation \cite{Breuer2002} of the form $d\rho(t)/dt=\mathcal{L}(t)\rho(t)$ that governs the evolution of the emitter density operator $\rho(t)$. The generator superoperator is $\mathcal{L}(t)=-i\mathcal{H}(t) + \gamma\mathcal{D}_\sigma$, where $\mathcal{H}\rho = [H, \rho]$ is the action of applying the commutator with the emitter Hamiltonian $H$, $\gamma$ is the emitter decay rate, and $\mathcal{D}_\sigma\rho = \sigma\rho\sigma^\dagger - \sigma^\dagger\sigma\rho/2-\rho\sigma^\dagger\sigma/2$ is the dissipation superoperator with $\sigma=\ket{g}\bra{e}$ being the emitter lower operator. The Hamiltonian is $H(t) = \Omega(t)\sigma_x/2$, where $\sigma_x = \sigma + \sigma^\dagger$ is the Pauli $\X$ operator and $\Omega(t)$ is the excitation pulse Rabi frequency.

To consider the ideal limit for this scheme, where all multi-photon components are fully suppressed, we define $\Omega(t)$ to be a perfect square input pulse beginning at $t=0$ and ending at $t=\tau$ as illustrated in Figure \ref{fig:tls}(a). That is, we define the Rabi frequency to be
\begin{equation}
    \Omega(t) = \left\{\begin{matrix}
        \frac{\Theta}{\tau} & 0\leq t \leq \tau\\
        0 & t > \tau
    \end{matrix}\right.
\end{equation}
For ideal light-matter coupling, the Rabi frequency is related to the average number of photons $|\alpha|^2$ by $|\alpha|^2=\Theta^2/4\gamma\tau$, where $\Theta=\int |\Omega(t)|^2 dt$ is the pulse area. After the pulse excites the two-level emitter, as illustrated in Figure \ref{fig:tls}(b), we collect only the light found in the output mode for when $t> \tau$ (see Figure \ref{fig:tls}(c)).
\footnote{For the correctness of our protocol we assume that no photon is emitted for $t \leq \tau$. If a photon is emitted during that interval and filtered out, it will always result in a perfectly mixed state at the output. However, the security is not impacted by this filtering since it will at most reduce the information available to Bob.}

\begin{figure}
    \centering
    \includegraphics[width=0.90\textwidth, trim=0 6 0 13, clip]{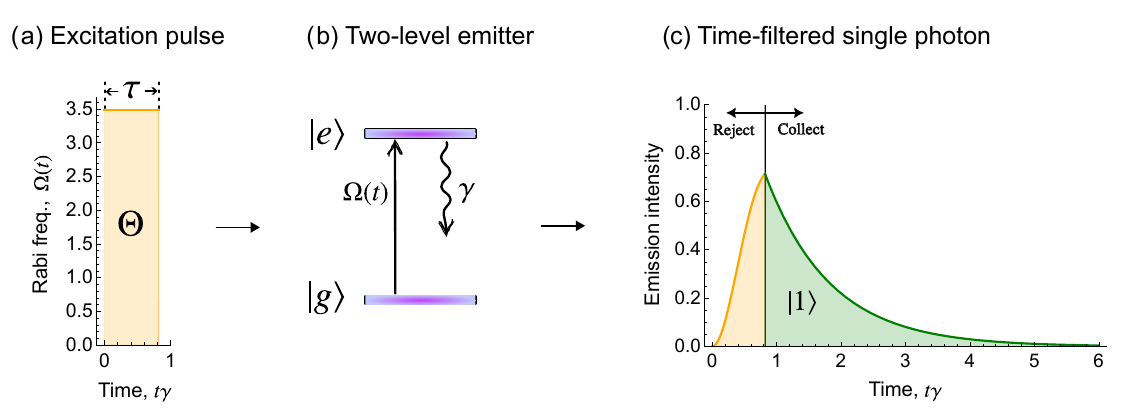}
    \caption{A resonantly driven two-level emitter for producing single photons from a weak coherent state. (a) The square excitation pulse with pulse area $\Theta$ and width $\tau = \Theta^2/(4\gamma|\alpha|^2)$ in a coherent state with average photon number $|\alpha|^2$. (b) A two-level emitter driven by the excitation pulse Rabi frequency $\Omega(t)$ and decaying with rate $\gamma$. (c) The emission intensity quantified by the occupation probability of the two-level emitter, highlighting the excitation and decay steps. By collecting emission after the pulse has finished interacting with the emitter, the emitted light is in a single photon state.}
    \label{fig:tls}
\end{figure}

Since for $t \geq \tau$, all multi-photon components will vanish, the single-photon emission probability $\eta_1$ is equal to the occupation probability of the emitter excited state at time $t=\tau$. This can be solved in a straightforward way by taking the matrix exponential of the generator $\eta_1 = \bra{e}e^{\tau\mathcal{L}}\ket{g}$, resulting in
\begin{equation}
    \eta_1 = 1 - \frac{1}{\gamma^2+2\Omega^2}\left[\gamma^2+\Omega^2+\Omega^2\left(\cos(\tau\Omega^\prime)+\frac{3\gamma}{\Omega^\prime}\sin(\tau\Omega^\prime)\right)e^{-3\gamma\tau/4}\right],
\end{equation}
where $\Omega = \Omega(\tau) = 4\gamma|\alpha|^2/\Theta$ and $\Omega^\prime=\sqrt{(\gamma/4)^2 + \Omega^2}$. 

To be sure that the two-level emitter can allow for security amplification, we must check if the probability of generating a single photon $\eta_1$ is greater than the probability of receiving two or more photons $p_{\alpha,2}$. That is, the performance of a generation scheme is determined by the maximum value of $\eta_1$ such that $\eta_1 - p_{\alpha,2}$ is positive.

In the limit $\gamma \ll \Omega$, where $|\alpha|^2\gg 1$, the expression for $\eta_1$ reduces to $\eta_1=\sin(\Theta/2)^2$, as expected from the pulse area theorem \cite{Fischer2017}. In this limit, a $\pi$ pulse corresponding to $\Theta=\pi$ is the optimal condition for excitation. 
However, this is no longer true when $\gamma\simeq \Omega$, which occurs when $|\alpha|^2$ is on the order of $1$. For example, numerically maximising $\eta_1$ over $\Theta$ for $|\alpha|^2=1$ implies $\eta_1=0.48$ using $\Theta=0.78\pi<\pi$. By plotting  $\eta_1$ maximised for over $\Theta$ as a function of $|\alpha|^2$ (Figure \ref{fig:tls_quality}), we can identify that $\eta_1=0.71$ is the upper limit on the efficiency for this single-photon generation scheme that would still allow for security amplification. This corresponds to an average photon number of $|\alpha|^2=2.5$ in the excitation pulse with width of $\tau=0.82/\gamma$ and area of $\Theta=0.91\pi$.

\begin{figure}[ht]
    \centering
    \includegraphics[width=0.50\textwidth]{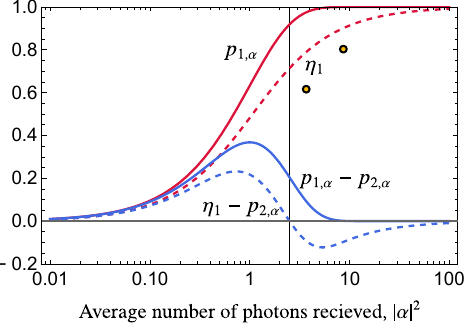}
    \caption{The single-photon efficiency $\eta_1$ (dashed red curve) and gap $\eta_1-p_{2,\alpha}$ (dashed blue curve) where a positive value $\eta_1 - p_{2,\alpha} > 0$ indicates security amplification is possible. The solid curves show the upper bounds on single-photon efficiency $p_{1,\alpha}$ and corresponding gap $p_{1,\alpha}-p_{2,\alpha}$ set by the coherent state photon number statistics.  For a two-level emitter driven by a square pulse, the maximum efficiency while allowing for security amplification is $\eta_1=0.71$ (vertical black line) corresponding to the average number of photons $|\alpha|^2=2.5$, pulse area $\Theta=0.91\pi$, and pulse width $\tau=0.82/\gamma$. The points indicate experimentally measured values from Ref.~\cite{Giesz2016}.}
    \label{fig:tls_quality}
\end{figure}

Note that one could reduce the amount of time filtering to increase $\eta_1$ at the cost of accepting more multi-photon cases that reduce the single-photon purity, which introduces an efficiency-fidelity trade-off. However, an accurate simulation of this case requires capturing the quantum fluctuations of the excitation pulse, which can be accomplished following a collision model approach \cite{Maffei2022} or a quantum cascaded interaction \cite{Kiilerich2019}.

\subsection{Saturating the upper-bound via alternative energy level structures}
\label{subsec:other-excitations}

The scheme analysed in the previous section gives a concrete idea of how to satisfy the criteria for security amplification using an excitation scheme that can be directly applied to protocols for deterministic graph state generation, but it is not the only approach to generating a single photon from a weak coherent state nor is it the most efficient. It may be possible to improve $\eta_1$ by using a cavity to control the light-matter coupling or by altering the shape of the coherent pulse. But it is not easy to see how to attain the upper bound $\eta_1 = p_{1,\alpha}$ using just a two-level emitter, or that it is even possible in general. Interestingly, it is possible to reach the upper bound by increasing the complexity of the emitter energy level structure.

\begin{figure}[ht]
    \centering
    \includegraphics[width=0.85\textwidth]{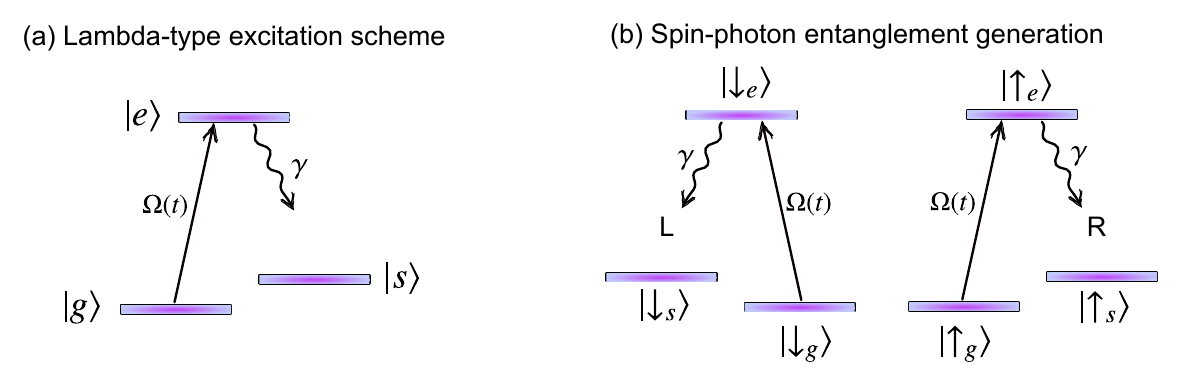}
    \caption{(a) Single-photon generation using a $\Lambda$-type emitter. The excitation and collection transitions are independent and distinguishable in resonant frequency. Excitation occurs by driving $\ket{g}$ to $\ket{e}$ with the Rabi frequency $\Omega(t)$. Emission from $\ket{e}$ to $\ket{s}$ dominates at the rate $\gamma$, decoupling the system from the laser as soon as a single photon is produced. (b) Application of the $\Lambda$-type excitation for producing spin-photon entanglement to generate photonic graph states of qubits encoded in left (L) and right (R) circular polarisation.}
    \label{fig:lambda}
\end{figure}

The upper bound $p_{1,\alpha}$ can be reached by considering a system with just one additional ground state $\ket{s}$ that is detuned from $\ket{g}$, so that we have a three-level $\Lambda$ type energy level structure (see Figure \ref{fig:lambda}(a)). If the $\Lambda$ system is initially in the ground state $\ket{g}$, a coherent pulse can be used to excite the system to state $\ket{e}$. As soon as a single photon is emitted by the system due to the decay of $\ket{e}$ to $\ket{s}$, the system is no longer coupled by the excitation pulse and thus cannot be re-excited. Therefore, by filtering the emission in frequency, the multi-photon emission probability can be fully suppressed without sacrificing the single-photon probability. This allows for moving arbitrarily close to the upper bound $p_{1,\alpha}$ by simply increasing the quality of mode matching between the excitation pulse and the excitation transition, and also by engineering the decay of the system so that the transition $\ket{e}$ to $\ket{s}$ is much more likely than the transition $\ket{e}$ to $\ket{g}$.

Although the $\Lambda$ scheme can be brought arbitrarily close to the upper bound $p_{1,\alpha}$ allowing for security amplification while asymptotically reaching perfect efficiency $\eta_1\rightarrow 1$, it is more difficult to work with and control physical systems that allows for two such $\Lambda$ configurations, which would be necessary to deterministically generate photonic graph states (see Figure \ref{fig:lambda}(b)). The excitation scheme also favours temporally long pulses, which may negatively impact computational rate. However, this level structure has already been used to generate entangled photonic graph states using a single rubidium atom in an optical cavity \cite{Thomas2022}.

\subsection{Experimental results paving the way for protocol implementations}
\label{subsec:experimental}

The excitation of quantum dots using weak laser pulses has already been realised in~\cite{Giesz2016}. They demonstrated that a coherent state with an average of $3.8$ ($8.6$) photons was able to invert the transition of a quantum dot exciton with a $62\%$ ($81\%$) probability. Comparing these values to the predictions from Figure~\ref{fig:tls_quality}, Section~\ref{subsec:two-lev}, indicates that these inversion probabilities are not yet sufficient to allow for security amplification, but the trend is favourable and shows that proof-of-concept experiments may be feasible with current technology. It is important to note however that in Ref.~\cite{Giesz2016} they use resonant excitation to maximise light-matter coupling and thus require polarisation filtering to separate the single photons from the excitation pulse. This pulse-filtering method is not compatible with the spin-polarisation entanglement which we use to generate graph states. Hence further developments and experimental proposals are required.

The Lindner-Rudolph protocol for generating spin-polarisation entanglement using quantum dots has been experimentally demonstrated in Refs.~\cite{Thomas2022, Cogan2023, Coste2023, Meng2023}. For example, in Ref.~\cite{Coste2023}, three horizontally-polarised laser pulses are used to generate a spin-photon-photon linear cluster state. A similar protocol and setup, adding only a controlled polarisation to the second pulse, can be used to perform a proof-of-concept experiment to confirm that the emission can indeed be described by the operator $E_{\rm qe}(\theta)$. However, this technique uses high-energy off-resonant excitation pulses \cite{Thomas2021}, meaning that the laser power is too high to amplify security and recover blindness.

Many approaches could be used to reduce excitation power while still producing entangled states of light. As demonstrated theoretically in Section \ref{subsec:two-lev}, a resonant excitation with a fast laser pulse and temporal filtering can, in principle, allow for security amplification. Furthermore, in the schemes described above, the excitation pulse must pass through a highly-reflective cavity mirror before interacting with the quantum emitter. Using multi-frequency excitation pulses as in the SUPER scheme~\cite{Bracht2021,Karli2022}, or exciting the source along a spatial axis which is different from the predominant emission direction, could by-pass this limitation and thus reduce the required pump power. This second approach would also allow for spatial filtering. Finally, as discussed in Section~\ref{subsec:other-excitations}, emitters with a more complex energy level structure would open new avenues to engineer energetically efficient deterministic entanglement generation schemes.

\section{Conclusion and future work}
\label{sec:dicuss}

In this paper, we presented a protocol that greatly reduces the technological requirements for performing blind and secure delegated quantum computations. Not only does the client no longer require a single qubit source or measurement device, but the server does not need to perform complex operations such as measuring the photon number without disturbing the photons' polarisation, or performing photon-photon entangling gates. All that the server requires is interacting quantum emitters which can generate spin-photon entanglement. Our protocol is information-theoretically secure, as proven in the composable security framework of Abstract Cryptography. We have also demonstrated that there are no fundamental limitations to the implementation of our protocol by analysing the emission probability of a two-level quantum emitter and comparing it to the thresholds required for correctness and security amplifications. We find that existing experimental results approach these values, indicating that our protocol could feasibly be implemented in the future.

However, the security analysis relies on very pessimistic assumptions. In particular, we assume that all information about an angle is leaked so long as a laser pulse contains more than one photon. This requires the server to isolate the photons contained in a pulse and manipulate them individually, which is challenging in itself, but it also overestimates the amount of information that the server can gain from these states. Relaxing this assumption or redefining security in a less drastic way (in terms of distinguishing advantage, or including a more comprehensive study of leakage) would yield an improved security bound. Another path would be to randomise the states even more.

Then, in the current work, we use 2D graph structures which are sufficient for universal blind quantum computing. Expanding these results to 3D graphs which also enables fault-tolerance would be an interesting direction to handle both photon losses and the server's honest processing errors~\cite{Morimae2012}. 

Finally, our analysis in Section~\ref{subsec:other-excitations} leaves a very interesting open question: which physical systems and practical excitation schemes can allow for energetically efficient deterministic photonic graph state generation to optimise security amplification? Although a definitive answer to this question goes beyond the scope of this work, we believe that frameworks such as microscopic quantum energetics \cite{Weimer2008, Hossein2015, Maffei2021} can be readily applied to solve this problem. 

\subsection*{Acknowledgements}
This work has been co-funded by the European Commission as part of the EIC accelerator program under the grant agreement 190188855 for SEPOQC project, and by the Horizon-CL4 program under the grant agreement 101135288 for EPIQUE project.

\bibliographystyle{quantum}
\bibliography{bib}

\appendix

\section{Additional preliminaries}
\label{app:prelim}

\subsection{Formal description of the blind and secure delegated protocols from Section \ref{subsec:ubqc}}
\label{app:dqc-protocols}

\begin{protocol}[ht]
\caption{Classical input UBQC}
\label{prot:UBQC}
\begin{algorithmic} [0]

\STATE \textbf{Public information:} Description of a graph $G = (V, E, I , O)$ with input vertices $I \subset V$ and output vertices $O \subset V$ and a measurement order $\preceq$ over vertices.

\STATE \textbf{Inputs:} Alice inputs a bit-string $x \in \{0, 1\}^{|I|}$, base measurement angles $\{\phi_v\}_{v \in V}$ and a flow $f$ on the graph $G$.

\STATE \textbf{Protocol:}

\begin{enumerate}
\item Alice, for each vertex $v \in V$, samples uniformly at random $\theta_v \in \Phi$ and sends $\ket{+_{\theta_v}}$ to Bob.
\item Bob receives the qubits one-by-one and applies the entangling operations $\CZ$ that correspond to the edges of the graph $G$.
\item For each qubit $v \in V$, following the partial order of the flow:
\begin{enumerate}
\item Alice samples uniformly at random $r_v \in \{0, 1\}$, calculates and sends the following measurement angle to Bob:
\begin{equation}
\delta_v = \phi'_v + \theta_v + r_v\pi + x_v\pi.
\end{equation} 
\item Bob measures in the $\ket{\pm_{\delta_v}}$ basis and returns to Alice the outcome $b_v$. Alice sets $s_v = b_v \oplus r_v$.
\end{enumerate}
\item Alice sets the bit-string $\{s_v\}_{v \in O}$ as her output.
\end{enumerate}

\end{algorithmic}
\end{protocol}

\begin{protocol}[ht]
\caption{SDQC for classical input and output $\mathsf{BQP}$ computations via UBQC repetitions}
\label{prot:SDQC}
\begin{algorithmic} [0]

\STATE \textbf{Public information:}
\begin{itemize}
\item Description of a graph $G = (V, E, I , O)$ with input vertices $I$ and output vertices $O$
\item Measurement order $\preceq$ over vertices
\item Number of repetitions $N$ and tests $t$, threshold of failed tests $w$
\item Set of stabilisers tests $\mathfrak{T}$ as described in Section~\ref{subsubsec:sdqc}
\end{itemize}

\STATE \textbf{Inputs:} Alice inputs a bit-string $x \in \{0, 1\}^{|I|}$, base measurement angles $\{\phi_v\}_{v \in V}$ and a flow $f$ on the graph $G$. We assume that the honest computation output $\mathsf{C}(\dyad{x})$, where $\mathsf{C}$ is the CPTP map defined by the measurement pattern $(G, \{\phi_v\}_{v \in V}, f, \preceq)$, follows a probability distribution which produces one bit string with probability $\frac{1}{2} + p$, for constant $p$.

\STATE \textbf{Protocol:}

\begin{enumerate}
\item Alice samples uniformly at random a set of test runs $T \subset \{1, \ldots, N\}$ of size $t$.
\item Alice and Bob run $N$ times an instance of the UBQC Protocol \ref{prot:UBQC}. For instance $i$:
\begin{itemize}
\item if $i \notin T$, Alice performs the computation described by $(G, \{\phi_v\}_{v \in V}, x, f, \preceq)$;
\item otherwise, Alice samples uniformly at random a test $\mathsf{T}$ from the set $\mathfrak{T}$ and instructs Bob to measure the qubits according to the test, in the order $\preceq$. Alice chooses at random the angle of measurement of qubits unconstrained by the test. The test passes if the parity of measurement outcomes of tested qubits is $0$ and it fails otherwise.
\end{itemize}
\item Alice checks if the number of failed tests is larger than the threshold $w$ and aborts by setting $\Ab$ as her output if it is.
\item Otherwise Alice sets as her output the bit string appearing most often among the outputs of computation rounds $i \notin T$.
\end{enumerate}

\end{algorithmic}
\end{protocol}

\subsection{Abstract cryptography framework}
\label{app:ac}
In this section, we describe the framework in which we will prove the security of our protocol. We focus on the case of two-party protocols with one-sided statistical security and refer to~\cite{Maurer2011,Dunjko2014} for a more complete formulation.

\subsubsection{Framework description}

In this framework, the purpose of a protocol $\pi$ is, given a number of available resources $\mathcal{R}$, to construct a new resource.
This new resource can itself be reused in a future protocol, which means that any protocol proven secure in this framework is inherently composable.

The actions of honest players in a given protocol are represented as a sequence of efficient CPTP maps acting on their internal quantum registers -- which may contain communication registers, both classical and quantum.  
A two-party quantum protocol is therefore described by $\pi = (\pi_A, \pi_B)$ where $\pi_X$ is the aforementioned sequence of efficient CPTP maps executed by party $X$, called the \emph{converter} of party $X$.

A \emph{resource} $\mathcal{R}$ is described as a sequence of CPTP maps with an internal state.  
It has input and output interfaces describing which parties may interact with it.  
It works by (i) having the parties send a state at each of its input interfaces; (ii) applying the specified CPTP map after all input interfaces have been initialised and (iii) outputting the resulting state at its output interfaces in a specified order. These can be represented as in Figure~\ref{fig:ac-proto}.

\begin{figure}[ht]
\centering
\resizebox{\textwidth}{!}{\subfloat[A two-party resource with one input and one output interface per party. After the inputs have been sent, the resource applies a fixed computation and returns the outputs. All these can be classical or quantum.]{
\label{fig:ac-proto-a}
\includegraphics[height=4.5cm]{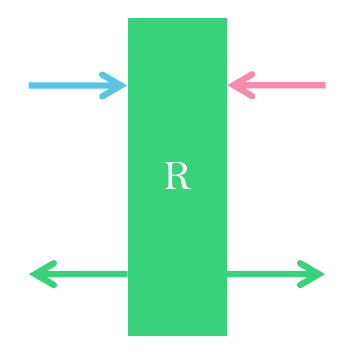}
}
\quad\quad
\subfloat[A protocol between two parties, each with their own converter labelled $\pi_A$ and $\pi_B$, which requires one call to resource $\mathcal{R}$. Each converter is a sequence of computations, quantum or classical. One computation in the sequence is applied in each round after receiving a communication from either the other party or a resource.]{
\label{fig:ac-proto-b}
\includegraphics[height=4.5cm]{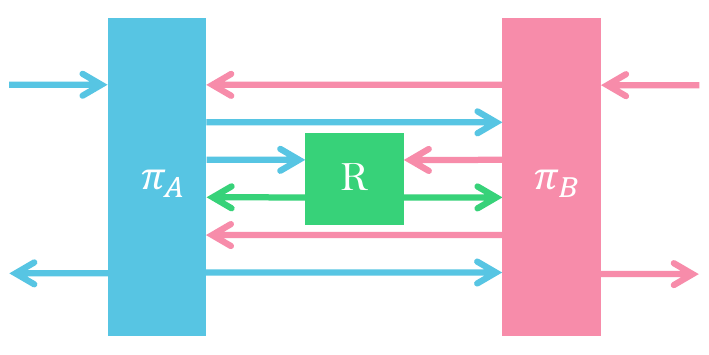}
}}
\caption{Basic building blocks of the Abstract Cryptography framework: converters representing the actions of protocol participants and resources representing existing secure functionalities.}
\label{fig:ac-proto}
\end{figure}

To define the security of a protocol, we first specify a pseudo-metric on the space of resources using a special type of converters called \emph{distinguishers}. These capture ``everything that can happen outside the protocol and interact with it''.
Its aim is to test how close the behaviours are of two resources $\mathcal{R}_0$ and $\mathcal{R}_1$ which have the same input and output interfaces, regardless of the context in which they are used.
It attaches to the inputs and outputs of one of the resources, interacts with it according to its own strategy and outputs a single bit indicating its guess as to which resource it had access to.
Two resources are said to be indistinguishable if no distinguisher can make this guess with good probability. 

\begin{definition}[Statistical indistinguishability of resources]
\label{def:indist-res}
Let $\epsilon > 0$ and $\mathcal{R}_0$ and $\mathcal{R}_1$ be two resources with same input and output interfaces.  
The resources are $\epsilon$-statistically-indistinguishable if, for all distinguishers $\mathcal{D}$, we have:

\begin{equation}
\Bigl\lvert\Pr[b = 1 \mid b \leftarrow \mathcal{D}\mathcal{R}_0] - \Pr[b = 1 \mid b \leftarrow \mathcal{D}\mathcal{R}_1]\Bigr\rvert \leq \epsilon
\end{equation}

We then write $\mathcal{R}_0 \!\underset{\epsilon}{\approx}\! \mathcal{R}_1$.

\end{definition}

The correctness of a protocol $\pi$ applied to resource $\mathcal{R}$ can be expressed as the indistinguishability between the resource $\pi_A \mathcal{R} \pi_B$ and a desired target resource $\mathcal{S}$. This can be represented as in Figure~\ref{fig:ac-cor}.

\begin{figure}[ht]
\centering
\resizebox{\textwidth}{!}{\subfloat[The constructed resource $\pi_A \mathcal{R} \pi_B$. Both participants execute honestly the actions as prescribed by the protocol, represented by both converters $\pi_A,\pi_B$ being applied to resource $\mathcal{R}$.]{
\label{fig:ac-cor-a}
\includegraphics[height=5cm]{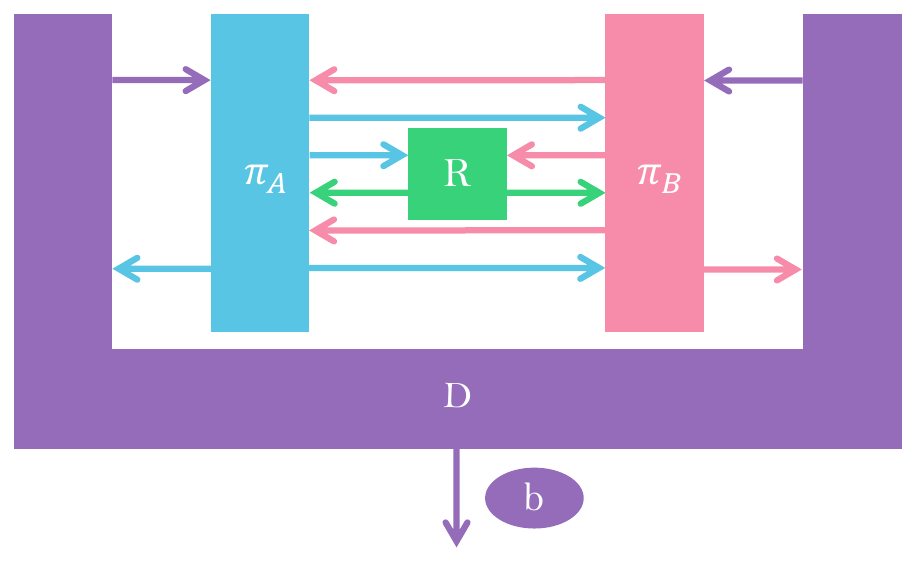}
}
\quad\quad
\subfloat[The desired target resource $\mathcal{S}$.]{
\label{fig:ac-cor-b}
\includegraphics[height=5cm]{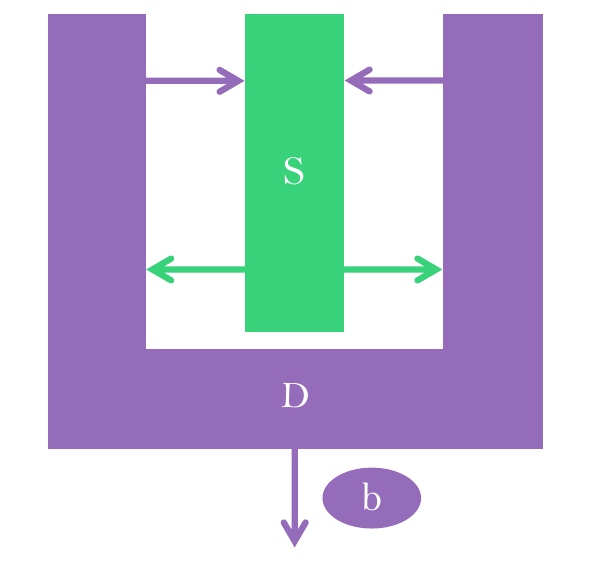}
}}
\caption{Correctness: the distinguisher tries to tell apart the real protocol and the resource which the protocol attempts to construct. The distinguisher has access only to the input and output interfaces of both participants.}
\label{fig:ac-cor}
\end{figure}

A corrupted Bob is allowed to use any other CPTP maps instead of the honest converter specified by the protocol.
This is modelled by removing Bob's converter, in which case the constructed resource becomes $\pi_A\mathcal{R}$. On the other side, there must exist a converter called a \emph{simulator} $\sigma_B$ which attaches to Bob's interfaces of $\mathcal{S}$ and aims to reproduce the transcript of honest Alice interacting with corrupted Bob. The simulator also controls any resource which the protocol uses. The distinguisher then tries to tell apart $\pi_A\mathcal{R}$ and $\mathcal{S}\sigma_B$. This is shown in Figure~\ref{fig:ac-sec}.

\begin{figure}[ht]
\centering
\resizebox{\textwidth}{!}{\subfloat[The resource $\pi_A \mathcal{R}$ with corrupted Bob. Only Alice's converter is applied, and the distinguisher can input any valid message at Bob's input interfaces. It also recovers all of Bob's incoming messages]{
\label{fig:ac-sec-a}
\includegraphics[height=4.5cm]{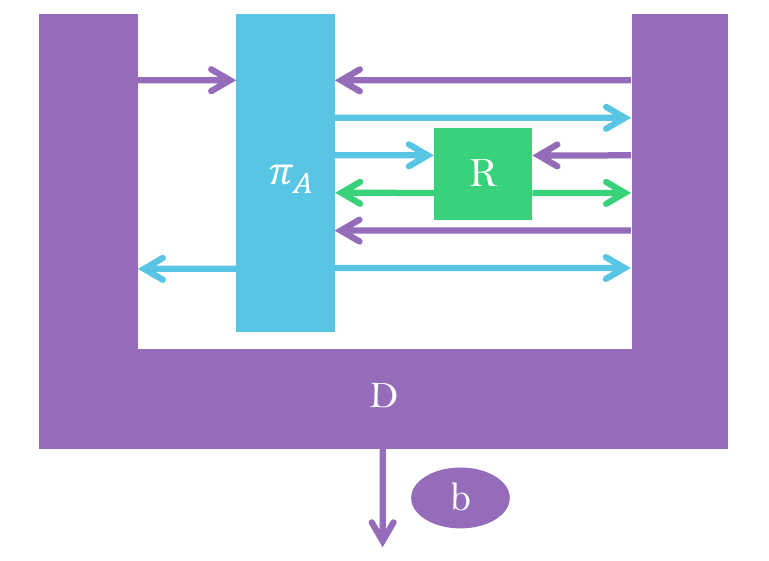}
}
\quad\quad
\subfloat[The resource $\mathcal{S}$ together with simulator $\sigma_B$. The simulator receives all of Alice's incoming messages and tries to simulate the messages which Alice would have sent to Bob in response. It has access to Bob's interfaces to resource $\mathcal{S}$, to which it must provide an input, and controls the behaviour of resource $\mathcal{R}$.]{
\label{fig:ac-sec-b}
\includegraphics[height=4.5cm]{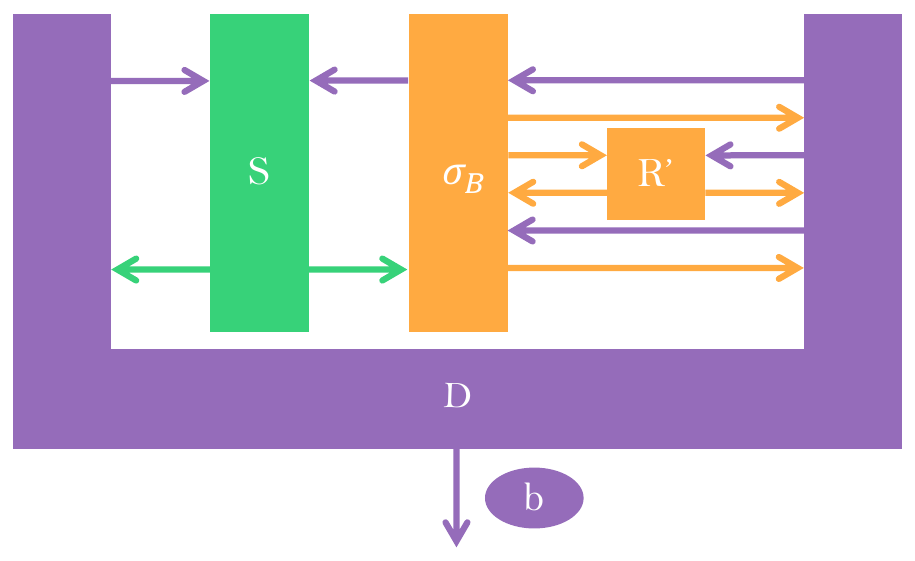}
}}
\caption{Security: the distinguisher can now feed any message of its choice to Alice or the simulator in the execution and aims to distinguish these scenarii. It now has access to the full transcript as well as Alice's input and output interfaces.}
\label{fig:ac-sec}
\end{figure}

Intuitively, if the resources $\pi_A\mathcal{R}$ and $\mathcal{S}\sigma_B$ are indistinguishable, it means that the transcript received by Bob during the execution of the protocol does not give more information about Alice's input and output than what Bob can already learn based on their own input and output (which are the only things that the simulator has access to via the resource $\mathcal{S}$).
Both the correctness and security are formalised together in Definition \ref{def:const-sec} below.

\begin{definition}[Statistical construction of resources]
\label{def:const-sec}

Let $\epsilon > 0$.  
We say that a two-party protocol $\pi$ $\epsilon$-constructs resource $\mathcal{S}$ from resource $\mathcal{R}$ against adversarial Bob if:
\begin{enumerate}
\item It is correct: $\pi \mathcal{R} \!\underset{\epsilon}{\approx}\! \mathcal{S}$;
\item It is secure against Bob: there exists a simulator $\sigma_B$ such that $\pi_A\mathcal{R} \!\underset{\epsilon}{\approx}\! \mathcal{S} \sigma_B$.
\end{enumerate}

\end{definition}

The security of a protocol guarantees that it can be replaced by the resource it constructs at a fixed cost of $\epsilon$ regardless of the context. We say that the construction is \emph{perfect} if $\epsilon = 0$.

If multiple resources are replaced, the total cost is the sum of the individual replacements. This means that secure protocols can be safely composed sequentially or concurrently as expressed in the General Composition Theorem below (Theorem 1 from~\cite{Maurer2011}).

\begin{theorem}[General composition of resources]
\label{thm:compos}

Let $\mathcal{R}$, $\mathcal{S}$ and $\mathcal{T}$ be resources, $\alpha$, $\beta$ and $\mathsf{id}$ be protocols (where protocol $\mathsf{id}$ does not modify the resource it is applied to). Let $\circ$ and $\mid$ denote respectively the sequential and parallel composition of protocols and resources. Then the following implications hold:
\begin{itemize}
\item The protocols are \emph{sequentially composable}: if $\alpha \mathcal{R} \!\underset{\epsilon_{\alpha}}{\approx}\! \mathcal{S}$ and $\beta \mathcal{S} \!\underset{\epsilon_{\beta}}{\approx}\! \mathcal{T}$ then $(\beta \circ \alpha) \mathcal{R} \!\!\!\underset{\epsilon_{\alpha} + \epsilon_{\beta}}{\approx}\!\!\! \mathcal{T}$
\item The protocols are \emph{context-insensitive}: if $\alpha \mathcal{R} \!\underset{\epsilon_{\alpha}}{\approx}\! \mathcal{S}$ then $(\alpha \mid \mathsf{id}) (\mathcal{R} \mid \mathcal{T}) \!\underset{\epsilon_{\alpha}}{\approx}\! (\mathcal{S} \mid \mathcal{T})$
\end{itemize}
\end{theorem}

Combining the two properties presented above yields concurrent composability (the distinguishing advantage cumulates additively as well). 

Note that the AC framework is equivalent to the Quantum Universal Composability (Q-UC) Model of~\cite{Unruh2010} if a single Adversary controls all corrupted parties -- which is the case in this work. Therefore any protocol which is secure in the Q-UC model is also secure in the AC model considered here.

\subsubsection{Useful resources and security results}

We will describe in this section three resources, the first being used to construct the second and third via Protocols \ref{prot:UBQC} and \ref{prot:SDQC} respectively, for which we give the security results for completeness.

Resource \ref{res:rsp} (Remote State Preparation or RSP) allows Alice to prepare a quantum state on a device held by Bob. A protocol constructing this resource essentially replaces a direct quantum channel between Alice and Bob for a subset of quantum states. We focus on the states used in the UBQC protocol, i.e.~blind graph states $\ket{G(\vec \theta)}$.

\begin{resource}[ht]
\caption{Remote Blind Graph State Preparation}
\label{res:rsp}

\begin{algorithmic}[0]

\STATE \textbf{Inputs:} Alice inputs a graph $G = (V, E)$ and a set of angles $\vec \theta \in \Phi^{\abs{V}}$.

\STATE \textbf{Computation by the resource:} The resource prepares and sends the state $\ket{G(\vec \theta)}$ to Bob, along with the description of the graph $G$.

\end{algorithmic}
\end{resource}

This resource is usually constructed by having Alice send one $\ket{+_\theta}$ state for each vertex $V$ in the graph, while Bob performs the entangling operations.

The second resource allows Alice to perform a quantum computation on Bob's device while at the same time controlling the amount of information that is leaked. However Bob can alter the outcome of the computation in any way that they desire and Alice cannot verify that Bob acted as they should have during the protocol. We describe this resource in the special case where the computation is expressed as an MBQC pattern. Note that the Alice's inputs and outputs are classical here because all qubits are measured.\footnote{The UBQC protocol also provides blindness guarantees if Alice has quantum inputs and outputs but we will not require this additional capability in this work.}

\begin{resource}[ht]
\caption{Blind Delegated Quantum Computation}
\label{res:bqc}
\begin{algorithmic}[0]

\STATE \textbf{Inputs:} 
\begin{itemize}
\item Alice inputs the description of a graph $G = (V, E, I , O)$ with input vertices $I \subset V$ and output vertices $O \subset V$, a measurement order $\preceq$ over vertices, a bit-string $x \in \{0, 1\}^{|I|}$, base measurement angles $\{\phi_v\}_{v \in V}$ and a flow $f$ on the graph $G$.
\item Bob chooses whether or not to deviate via two control bits $(e, c)$ (set to $0$ for honest behaviour). If $c = 1$, Bob has an additional input CPTP map $\mathsf F$ and state $\rho_B$.
\end{itemize}

\STATE \textbf{Computation by the resource:}
\begin{enumerate}
\item If $e = 1$, the resource sends $V, E, \preceq$ to Bob's interface.
\item If $c = 0$, it outputs $\mathcal{M}_O (\mathsf{U}(\dyad{x}))$ at Alice's output interface, where $\mathsf{U}$ is the unitary defined by the measurement pattern $(G, \{\phi_v\}_{v \in V \setminus O}, f, \preceq)$ before the measurements on the output qubits $O$, and $\mathcal{M}_O$ performs the measurements defined by $\{\phi_v\}_{v \in O}$ and outputs the outcome. Otherwise, it waits for the additional input and outputs $\mathcal{M}_C (\Tr_B(\mathsf F(\dyad{x}\otimes\rho_B)))$ at Alice's interface, where $\mathcal{M}_C$ is a computational basis measurement.
\end{enumerate}

\end{algorithmic}
\end{resource}

We can relate these first two resources if we replace the first step of the UBQC protocol (Alice sending the rotated qubits to Bob) by a call to the RSP Resource~\ref{res:rsp}. The following theorem from~\cite{Dunjko2014} then captures the security guarantees of the UBQC Protocol \ref{prot:UBQC} in the Abstract Cryptography Framework.

\begin{theorem}[Security of Universal Blind Quantum Computation]
\label{thm:sec-ubqc}
The UBQC Protocol \ref{prot:UBQC} perfectly constructs the Blind Delegated Quantum Computation Resource~\ref{res:bqc} from one instance of the Remote Blind Graph State Preparation Resource~\ref{res:rsp}.
\end{theorem}

The final resource allows Alice to perform a computation on Bob's device while both controlling its leakage and ensuring that it has been performed correctly. Here Bob can only decide whether Alice should receive the correct outcome or abort. We describe it again in terms of MBQC.

\begin{resource}[ht]
\caption{Secure Delegated Quantum Computation}
\label{res:sqc}
\begin{algorithmic}[0]

\STATE \textbf{Inputs:} 
\begin{itemize}
\item Alice inputs the description of a graph $G = (V, E, I , O)$ with input vertices $I$ and output vertices $O$, a measurement order $\preceq$ over vertices, a bit-string $x \in \{0, 1\}^{|I|}$, base measurement angles $\{\phi_v\}_{v \in V}$ and a flow $f$ on the graph $G$.
\item Bob chooses whether or not to deviate via two control bits $(e, c)$ (set to $0$ for honest behaviour).
\end{itemize}

\STATE \textbf{Computation by the resource:}
\begin{enumerate}
\item If $e = 1$, the resource sends $V, E, \preceq$ to Bob's interface.
\item If $c = 0$, it outputs $\mathcal{M}_O (\mathsf{U}(\dyad{x}))$ at Alice's output interface, defined as in Resource~\ref{res:bqc}. Otherwise, if $c = 1$, it outputs $\Ab$ at Alice's interface.
\end{enumerate}

\end{algorithmic}
\end{resource}

We can again relate the first and third resources via Protocol \ref{prot:SDQC} and the following theorem from \cite{Kapourniotis2023}.

\begin{theorem}[Security of Secure Delegated Quantum Computation]
\label{thm:sec-sdqc}
Let $N$ be the number of executions of the UBQC Protocol \ref{prot:UBQC} in Protocol \ref{prot:SDQC}. Then there exists values of $c$ and $w$ such that Protocol \ref{prot:SDQC} executed using these parameters $\epsilon_{\mathit{SDQC}}$-constructs the Secure Delegated Quantum Computation Resource~\ref{res:sqc} for classical input and output $\mathsf{BQP}$ computations from $N$ instances of the Remote Blind Graph State Preparation Resource~\ref{res:rsp}, for $\epsilon_{\mathit{SDQC}}$ exponentially decreasing in $N$.
\end{theorem}
 
\section{Proof of Theorem~\ref{thm:sec-graph-rsp}}
\label{app:sec-proof-graph}

We prove here the following security result from Section~\ref{subsec:blind-graph-rsp}.

\newcounter{tempresult}
\setcounter{tempresult}{\value{theorem}}
\setcounter{theorem}{\value{count:sec-graph}-1}
\begin{theorem}[AC security of Protocol~\ref{prot:graph-rsp}]
Protocol~\ref{prot:graph-rsp} perfectly constructs in the Abstract Cryptography framework the Blind Graph RSP Resource~\ref{res:rsp} from $\abs{V}$ instances of the Blind Graph State Extender Resource~\ref{res:graph-extend}.
\end{theorem}
\setcounter{theorem}{\value{tempresult}}

We prove separately the correctness in case of honest Bob, and the security against malicious Bob.

\begin{proof}[Proof of correctness]

We prove the result by showing that the desired blind graph state grows as desired during the iteration over the vertices in the order $<$ specified by Bob. We assume that at the start of the step generating vertex $v$ using quantum emitter $q$, the state in Bob's possession is equal to:
\begin{equation}
\CZ_{q, v'}\ket{+}_q\ket{\tilde{G}_{< v}(\vec \theta)},
\end{equation}
where $v'$ is the vertex generated using quantum emitter $q$ right before $v$, i.e.~$v' \in V_q$ with $v' < v$ and $(v, v') \in E$, and $\ket{\tilde{G}_{< v}(\vec \theta)}$ is a redundantly-encoded blind graph state associated to the subgraph induced by the subset of vertices $V_{< v} = \{w \in V \mid w < v\}$. This is true at the start of the protocol, where the set of previously generated vertices is empty and all quantum emitters are initialised in state $\ket{+}$.

After calling the Blind Graph Extender Resource, the state in Bob's hands is equal to
\begin{equation}
\CZ_{q, v'}(\ket{00}_{q, v} + e^{i(-1)^{b_v} \theta_v}\ket{11}_{q, v})\ket{\tilde{G}_{< v}(\vec \theta)}.
\end{equation}
Bob then applies the corrections $\X^{b_v}$ to qubits $q$ and $v$ and $\Z^{b_v}$ to qubit $v'$. Commuting the operator $\X^{b_v}$ on qubit $q$ through the $\CZ$ cancels out the $\Z^{b_v}$ correction, yielding the state:
\begin{equation}
\CZ_{q, v'}(\ket{00}_{q, v} + e^{i\theta_v}\ket{11}_{q, v})\ket{\tilde{G}_{< v}(\vec \theta)}.
\end{equation}
Bob then generates $n_v$ additional qubits by applying the emission operator to quantum emitter $q$, for an $n_v$ of his choice:
\begin{equation}
\CZ_{q, v'}(\ket{0}_{q, v}^{\otimes n_v} + e^{i\theta_v}\ket{1}_{q, v}^{\otimes n_v})\ket{\tilde{G}_{< v}(\vec \theta)}.
\end{equation}
He performs entanglement operations between $v$ and all qubits $w \in V_{< v} \setminus \{v'\}$ such that $(v, w) \in E$, which yields the following state for a value $n'_v \leq n_v$:
\begin{equation}
\prod_{\substack{(v, w) \in E \\ w < i}}\CZ_{v, w}(\ket{0}_{q, v}^{\otimes n'_v} + e^{i\theta_v}\ket{1}_{q, v}^{\otimes n'_v})\ket{\tilde{G}_{< v}(\vec \theta)}.
\end{equation}
Bob then applies a Hadamard gate to quantum emitter $q$. If there exists $v'' \in V_q$ such that $v < v''$, this moves the quantum emitter $q$ to the next vertex $v''$ as follows:
\begin{equation}
\CZ_{q, v}\ket{+}_q\ket{\tilde{G}_{< v''}(\vec \theta)}.
\end{equation}
If such a vertex does not exist, the subgraph then contains all vertices which should have been generated by quantum emitter $q$. In both cases, the redundantly-encoded graph state has increased in size by integrating vertex $v$ and all links between $v$ and previously generated vertices.

At the end of the vertex generation process, the state in Bob's register is therefore:
\begin{equation}
\prod_q \CZ_{q, v_q}\ket{+}_q\ket{\tilde{G}(\vec \theta)},
\end{equation}
where $v_q$ is the final qubit in the list $V_q$. Bob then measures each spin qubit $q$ in the computational basis, with outcome $c_q$. The effect of measurement outcome $1$ is removed by applying the correction $\Z$ on a qubit of vertex $v_q$. The result is $\ket{\tilde{G}(\vec \theta)}$, which is redundantly encoded. This redundancy is removed by measuring all qubits but one in each vertex in the $\ket{\pm}$ basis, recording the parity of the sum $d_v$ and applying $Z^{d_v}$ to the remaining qubit in vertex $v$. The result is $\ket{G(\vec \theta)}$, which concludes the proof.

\end{proof}

\begin{proof}[Proof of security]

To prove security against malicious Bob, we construct a simulator which has a single oracle access to the Blind Graph RSP resource and must emulate the communication of Alice without knowing which blind graph state Alice has chosen to send. We then show that this exchange is perfectly indistinguishable to any unbounded distinguisher from a communication with an honest Alice.

The simulator is described in Simulator~\ref{sim:mal-bob-graph} below.

\begin{simulator}[ht]
\caption{Malicious Bob}
\label{sim:mal-bob-graph}

\begin{enumerate}
\item The simulator performs the call to the Blind Graph RSP Resource~\ref{res:rsp} and receives a state $\ket{G(\vec \theta)}$ for an unknown value of $\vec \theta$, and the description of the graph $G = (V, E)$.
\item It applies $\prod_{(v, w) \in E} \CZ_{v, w}$ to the state $\ket{G(\vec \theta)}$, thus recovering the product state $\bigotimes_{v \in V} \ket{+_{\theta_v}}$.
\item For each vertex $v$ in the graph, the simulator emulates the behaviour of the Blind Graph State Extender Resource~\ref{res:graph-extend} as follows:
\begin{enumerate}
\item It receives a single-qubit state from Bob.
\item It performs a $\CNOT$ gate between Bob's qubit and the qubit associated to vertex $v$ in the state received from the Blind Graph RSP Resource.
\item It measures the qubit received from the Blind Graph RSP Resource in the computational basis, let $b_{s, v}$ be the measurement outcome.
\item It initialises a qubit in state $\ket{0}$ and applies a $\CNOT$ to Bob's qubit and this newly created qubit.
\item It sends these two qubits to Bob along with bit $b_{s, v}$.
\end{enumerate}
\item The simulator then stops.
\end{enumerate}

\end{simulator}

We now show that this interaction is perfectly equivalent to the one with honest Alice during an execution of the protocol with malicious Bob. The only interaction happens through the Blind Graph State Extender Resource~\ref{res:graph-extend} and we prove that the simulator defined above perfectly emulates its behaviour. We do this by analysing the input/output relation and show that the produced states are equal.

We denote $\rho_{qe, v}$ the state supplied by Bob to the resource at the $v$\textsuperscript{th} execution. In the real case the output state is $\RZ((-1)^{b_v}\theta_v)\CNOT(\rho_{qe, v}\otimes\op{0})$, together with classical bit $b_v$. We show that the simulation yields the same state in the case of pure states, which then generalises to mixed states. We write $\ket{\psi_v}_{qe} = \alpha_v\ket{0} + \beta_v\ket{1}$ for Bob's input. After step 3b the state can be written as:
\begin{equation}
\CNOT\ket{\psi}_{qe, v}\ket{+_{\theta_v}} = (\alpha_v\ket{0} + \beta_v e^{i\theta_v}\ket{1})\ket{0} + (\alpha_v e^{i\theta_v}\ket{0} + \beta_v\ket{1})\ket{1}.
\end{equation}
Therefore, after obtaining measurement result $b_{s, v}$ on the second qubit, the simulator has in its possession the state $\alpha_v\ket{0} + \beta_v e^{i(-1)^{b_{s, v}}\theta_v}\ket{1}$. After step 3d, the state sent by the simulator is then:
\begin{equation}
\alpha_v\ket{0}\ket{0} + \beta_v e^{i(-1)^{b_{s, v}}\theta_v}\ket{1}\ket{1},
\end{equation}
which is exactly identical to $\RZ((-1)^{b_v}\theta_v)\CNOT\ket{\psi_v}_{qe}\ket{0}$. This concludes the proof.
\end{proof}

\section{Proof of Theorem~\ref{thm:sec-gadget}}
\label{app:sec-proof}

We prove here the following security result from Section \ref{subsec:amplification}. 

\setcounter{tempresult}{\value{theorem}}
\setcounter{theorem}{\value{count:sec}-1}
\begin{theorem}[AC security of Protocol~\ref{prot:ghz-gadget}]
Let $\eta_1$ be the probability that the quantum emitter generates a photon after receiving a laser pulse of intensity $\abs{\alpha}^2$ and let $p_{\alpha, 2} = 1 - e^{-\abs{\alpha}^2} - \abs{\alpha}^2 e^{-\abs{\alpha}^2}$ be the probability that the laser pulse contains two or more photons. Let $(\alpha, n)$ be the public parameters used in Protocol~\ref{prot:ghz-gadget}, and let $t = \frac{\eta_1 + p_{\alpha, 2}}{2}n$. Then Protocol~\ref{prot:ghz-gadget} $\exp(-\nu_{\alpha}n)$-constructs in the Abstract Cryptography framework the Blind Graph State Extender Resource~\ref{res:graph-extend}, for $\nu_\alpha = (\eta_1 - p_{\alpha, 2})^2/2$.
\end{theorem}
\setcounter{theorem}{\value{tempresult}}

To do this we will require the following result for binomial distributions.

\begin{lemma}[Hoeffding's inequality for the binomial distribution] \label{lem:hoeffding_binomial}
	Let $X\sim\operatorname{Binomial}(n,p)$ be a binomially distributed random variable. For any $k \leq np$ it then holds that
	\begin{align*}
		\Pr \left[ X \leq k \right] \leq \exp \left( -2 (p-k/n)^2 n \right).
	\end{align*}
	Similarly, for any $k \geq np$ it holds that
	\begin{align*}
		\Pr \left[ X \geq k \right] \leq \exp \left( -2 (p-k/n)^2 n \right).
	\end{align*}
\end{lemma}

We start by proving two lemmas, analysing separately the correctness in case of honest Bob, and the security against malicious Bob. We then recombine both to prove the result of the theorem.

\begin{lemma}[Correctness with honest parties]
\label{lem:cor}
Let $\eta_1$ be the probability that the quantum emitter generates a photon after receiving a laser pulse of intensity $\abs{\alpha}^2$. Let $t \leq n \eta_1$, Protocol~\ref{prot:ghz-gadget} is $\epsilon_{cor}$-correct for $\epsilon_{cor} = \exp \left( -2 \left(\eta_1 - t/n \right)^2 n \right)$.
\end{lemma}
\begin{proof}

We first compute the probability that Alice has not aborted the protocol. We later show that in this case, the state is exactly the one desired by Alice up to local operations on Bob's side.

The protocol fails if and only if there is an abort by Alice, which happens if not enough pulses generated a qubit via the quantum emitter.\footnote{This is valid if we assume lossless components after preparation for the photons generated using Alice's pulses, which consist only of measurements in the $\X$ basis.}
We abuse notation and denote $\abs{S}$ the random variable counting the number of photons measured by Bob. Therefore we can define the correctness error as:
\begin{equation}
\epsilon_{cor} = \Pr[\Ab] = \Pr[\abs{S} \leq t]
\end{equation}
The probability that a photon has been emitted using one of Alice's pulses is $\eta_1$. 
The random variable $\abs{S}$ follows a binomial distribution of parameters $(\eta_1, n)$. For $t \leq n \eta_1$, we use Hoeffding's inequality (recalled in Lemma~\ref{lem:hoeffding_binomial}) to bound the abort probability as follows:
\begin{equation}
\epsilon_{cor} = \Pr[\abs{S} \leq t] \leq \exp \left( -2 \left(\eta_1 - t/n \right)^2 n \right).
\end{equation}

We now assume that $\abs{S} > t$. Bob inputs a state $\rho_{qe}$ and we must show that the output state of the protocol is the same as if the operator $\RZ(-(1)^{m_x}\theta)E_{\rm qe}$ had been applied to $\rho_{qe}$. We show this for a single-qubit pure state $\ket{\psi}_{qe} = \alpha\ket{0} + \beta\ket{1}$, and this will automatically generalise to mixed states.

Since Bob can choose the excitation scheme and power to emit the last photon, we assume that the final photon has been emitted as well. Based on the description of the effect of the emission operator $E_{\rm qe}$ in Section~\ref{subsec:graph-gen} together with $E_{\rm qe}(\theta_i) = \RZ(\theta_i)E_{\rm qe}$, we know that after step 2 of Protocol~\ref{prot:ghz-gadget}, the spin-photon state in Bob's device can be written as:
\begin{equation}
\RZ\left(\sum_{i \in S} \theta_i\right)\left(\alpha\ket{0}^{\abs{S}+2} + \beta\ket{1}^{\abs{S}+2}\right) = \RZ((-1)^{m_x}\theta - \bar{\theta})\left(\alpha\ket{0}^{\abs{S}+2} + \beta\ket{1}^{\abs{S}+2}\right),
\end{equation}
with the state including both the polarisation and the spin qubits and $\bar{\theta} = (-1)^{m_x}\theta - \sum_{i \in S}\theta_i$. The pulses that did not lead to a photon emission do not contribute to this state.

Measuring all polarisation qubits generated via pulses sent by Alice in the $\X$ basis with outcomes $b_i$ gives the following state for the spin and remaining polarisation encoded qubit (the one generated by Bob):
\begin{equation}
\alpha\ket{0}_{qe}\ket{0}_{ph} + (-1)^b e^{i(-1)^{m_x}\theta - i\bar{\theta}}\beta\ket{1}_{qe}\ket{1}_{ph},
\end{equation}
where $b = \bigoplus_{i \in S} b_i$. After applying the final correction $\RZ(\bar{\theta} + b\pi)$ to the polarisation qubit, Bob is in possession of the following final state:
\begin{equation}
\ket{\psi_\theta}_{qe, ph} := \alpha\ket{0}_{qe}\ket{0}_{ph} + e^{i(-1)^{m_x}\theta}\beta\ket{1}_{qe}\ket{1}_{ph},
\end{equation}
which is exactly the state we would have obtained from applying $\RZ(-(1)^{m_x}\theta)E_{\rm qe}$ to $\ket{\psi}$.

Starting with the mixed state $\rho_{qe}$, we therefore obtain the mixed state:
\begin{equation}
\rho_{cor} := \epsilon_{cor}\dyad{\Ab} + (1 - \epsilon_{cor})\op{m_x}\otimes\rho_{\theta, qe, ph},
\end{equation}
as Bob's outcome, where $\rho_{\theta, qe, ph} := \RZ(-(1)^{m_x}\theta)\CNOT\left(\rho_B\otimes\op{0}\right)$. This $\rho_{cor}$ state is $\epsilon_{cor}$-close to the state produced by the Blind Graph Extender resource, which concludes the proof.

\end{proof}

\begin{lemma}[Security against malicious Bob]
\label{lem:sec}
Let $p_{\alpha, 2} := 1 - e^{-\abs{\alpha}^2} - \abs{\alpha}^2 e^{-\abs{\alpha}^2}$ and $t \geq n p_{\alpha, 2}$, Protocol~\ref{prot:ghz-gadget} is $\epsilon_{sec}$-secure against malicious Bob for $\epsilon_{sec} = \exp \left( -2 \left(t/n - p_{\alpha, 2} \right)^2 n \right)$.
\end{lemma}
\begin{proof}

To prove the security against malicious Bob, we construct a simulator which has a single oracle access to the Blind Graph Extender resource and must emulate the communication of Alice without knowing Alice's input angle. We then show that this exchange is indistinguishable to any unbounded distinguisher from a communication with an honest Alice, up to a known failure probability.

The simulator is described in Simulator~\ref{sim:mal-bob} below.

\begin{simulator}[ht]
\caption{Malicious Bob}
\label{sim:mal-bob}

\begin{enumerate}
\item The simulator samples $n$ values $\{k_i\}_{1 \leq i \leq n}$ from the Poisson distribution with parameter $\abs{\alpha}^2$.
\item The simulator performs the call to the RSP Resource~\ref{res:rsp} and receives a state $\ket{+_\theta}$ for an unknown value of $\theta$.
\item It samples values $(\theta_1, \ldots, \theta_n)$ uniformly at random from angle set $\Phi$.
\item For each pulse $i$ which Alice sends to Bob in the real protocol:
\begin{itemize}
\item If $k_i \notin \bin$, the simulator sends $\ket{k_i}_{\theta_i}$;
\item If $k_i = 0$, it does not send anything;
\item If $k_i = 1$, it prepares an EPR-pair $\ket{\Psi} = \sfrac{1}{\sqrt{2}}(\ket{00}+\ket{11})$ and sends half to Bob.
\end{itemize}
\item The simulator then receives the set $S \subseteq \{1,\ldots,n-1\}$ of correctly generated photons. If $\abs{S} \leq t$ it sends $\Ab$ to Bob.
\item If there is no index $i \in S$ such that $k_i \in \bin$, the simulator outputs the error message $\err$ to the distinguisher and stops. Otherwise it continues. Let $s$ be an index such that $k_s = 1$, chosen by the simulator uniformly at random from suitable values. If no such index exists, it chooses $s$ such that $k_s = 0$.\footnotemark
\item For each index $i \in S \setminus \{s\}$ such that $k_i = 1$, the simulator measures its half of the associated EPR-pair in the basis $\{\dyad{+_{-\theta_i}}, \dyad{-_{-\theta_i}}\}$ and records the outcome $m_i$.
\item If $k_s = 1$, it applies $\RZ(\theta_s)$ to the state received from the RSP resource. It then teleports this state $\ket{+_{\theta + \theta_s}}$ using the EPR-pair associated to index $s$, by applying a $\CNOT$ gate between this qubit and the half EPR-pair and measuring the qubits in respectively the $\X$ and $\Z$ bases, with $m_{s, Z}, m_{s, X}$ the associated results. If $k_s = 0$, it does not do anything.
\item The simulator sends the correction angle $\bar{\theta} = -(-1)^{m_{s, X}}\theta_s - m_{s, Z}\pi - \sum_{i \in S \setminus \{s\}} \theta_i + m_i\pi$ and output bit $m_{s, X}$ to Bob.
\item The simulator then stops.
\end{enumerate}

\end{simulator}

\footnotetext{In the honest case, Bob never reports a value $i$ where $k_i = 0$ as being inside the set $S$ since these laser pulses do not excite the quantum emitter and yield no photons. However, malicious Bob may choose $S$ arbitrarily.}

It is clear that if the simulator sends the error message $\err$, the states are perfectly distinguishable since an honest Alice never sends this message in a real execution of the protocol.
Since the simulator has not aborted earlier, this means that the set $S$ is of size at least $t$. The simulator then sends the error message if all the pulses in set $S$ have more than two photons in them.
Let $E$ be the random variable counting the number of indices $i \in \{1,\ldots,n\}$ such that $k_i \geq 2$. We count all values since the distinguisher can freely choose the set $S$. Therefore we can define the security error as:
\begin{equation}
\epsilon_{sec} = \Pr[\err] = \Pr[E > t]
\end{equation}

The probability that there are at least two photons in a given laser pulse is:
\begin{equation}
\Pr[k \geq 2] = 1 - e^{-\abs{\alpha}^2} - \abs{\alpha}^2 e^{-\abs{\alpha}^2} = p_{\alpha, 2}.
\end{equation}
The random variable $E$ follows a binomial distribution of parameters $(p_{\alpha, 2}, n)$. This time for $t \geq n p_{\alpha, 2}$, we again use Hoeffding's inequality to bound the error probability:
\begin{equation}
\epsilon_{sec} = \Pr[E > t] \leq \Pr[E \geq t] \leq \exp \left( -2 \left(t/n - p_{\alpha, 2} \right)^2 n \right).
\end{equation}

We now show that if this error message is not sent, then the real and ideal states follow exactly the same distribution, meaning that no distinguisher can tell apart the simulation and the real execution. This must be done for each step of the protocol/simulation.

\paragraph{Before receiving the set $S$.}
In the protocol, Alice sends the states:
\begin{equation}
\rho_{\alpha,\theta_i} = e^{- \abs{\alpha}^2} \sum_{k=0}^\infty \frac{\abs{\alpha}^{2k}}{k!} \op{k}_{\theta_i}.
\end{equation}
 
The distinguisher can perform a non-destructive photon counting measurement to know how many states it has received in each pulse. The distinguisher then has in its possession the following states, for known values of $k_i$ and random $\theta_i$:
\begin{equation}
\bigotimes_{i = 1}^n \ket{k_i}_{\theta_i}.
\end{equation}

Let $I$ be the number of indices $i$ such that $k_i = 1$ in a given execution of the simulator. In the simulated case, the distinguisher receives, for random $\theta_i$ and with $\rho_{EPR} = \Tr_S(\ket{\Psi})$ being the half EPR-pair sent by the simulator:
\begin{equation}
\rho_{EPR}^{\otimes I}\bigotimes_{k_i \neq 1} \ket{k_i}_{\theta_i}.
\end{equation}

The values $k_i$ follows the same distribution in the real and simulated cases. These two sets of states are perfectly indistinguishable at this stage since for each $k_i = 1$, a single qubit $\ket{+_{\theta_i}} = \ket{1}_{\theta_i}$ and a half-EPR pair are both equal to the perfectly mixed state from the distinguisher's point of view.

\paragraph{Before sending the corrections.}
We now consider the state after the distinguisher sends the set $S$. In both cases, the states corresponding to indices not in $S$ can be discarded since they are not used in the rest of the protocol/simulation and yield therefore no additional information to the distinguisher.

If there is an index $i \in S$ such that $k_i = 0$, both the simulation and the real protocol produce the same output: the correctness of the protocol implies that the angle obtained by Alice and the state received by the distinguisher are both random and decorrelated (since all qubits from the set $S$ are required to reconstruct the correct state in both cases). From now on we consider that no such value appears in $S$.

If there are values $i \in S \setminus \{s\}$ such that $k_i = 1$, after the simulator's measurement of its half of the EPR-pair in the $\ket{\pm_{\theta_i}}$ basis, the states corresponding to these indices are indistinguishable from honestly generated $\ket{+_{\theta_i}} = \ket{1}_{\theta_i}$ states, up to an update of angle $\theta_i$ to account for the measurement result. We can then consider the worst case in which there is a single value $s \in S$ such that $k_s = 1$, meaning that for all $i \in S \setminus \{s\}$ we have $k_i \geq 2$. We can also assume that for all these other indices, the entire value of $\theta_i$ is leaked to the distinguisher. We can then replace the states $\bigotimes_{i \in S \setminus \{s\}} \ket{k_i}_{\theta_i}$ with $\bigotimes_{i \in S \setminus \{s\}} \ket{\theta_i}$, where the $\ket{\theta_i}$ are classical states. These values are identically distributed in both setting, being chosen uniformly at random from the set of angles $\Phi$. On top of that, right before Alice or the simulator send the corrections, the distinguisher has in its possession either $\ket{+_{\theta_s}}$ in the real case, or $\ket{+_{(-1)^{m_{s, X}}(\theta + \theta_s) + m_{s, Z}\pi}}$ in the ideal case. These are again perfectly indistinguishable due to the uniformly random $\theta_s$.

\paragraph{Final step.}
Finally, the distinguisher receives the corrections either from Alice or the simulator. In the real protocol we have $\bar{\theta} = (-1)^{m_x}\theta - \sum_{i \in S}\theta_i$, while in the simulation we have $\bar{\theta} = -(-1)^{m_{s, X}}\theta_s - m_{s, Z}\pi - \sum_{i \in S \setminus \{s\}} \theta_i$ (there are no values $m_i$ in this worst case setting). In the end, including the distinguisher's knowledge about Alice input angle $\theta$, we get:
\begin{align}
\mathit{Real:} &\ket{+_{\theta_s}}\ket{\theta}\ket{m_x}\ket{(-1)^{m_x}\theta - \theta_s - \sum_{i \in S \setminus \{s\}} \theta_i}\bigotimes_{i \in S \setminus \{s\}} \ket{\theta_i}, \\
\mathit{Sim:} &\ket{+_{(-1)^{m_{s, x}}(\theta+\theta_s) + m_{s, Z}\pi}}\ket{\theta}\ket{m_{s, x}}\ket{-(-1)^{m_{s,x}}\theta_s - m_{s,z}\pi - \sum_{i \in S \setminus \{s\}} \theta_i}\bigotimes_{i \in S \setminus \{s\}} \ket{\theta_i}.
\end{align}
In both cases, the distinguisher can use the values for $\theta_i$ (last values) it knows to remove them from the correction (next to last values). These values of $\theta_i$ then appear nowhere else in the states and are sampled from the same distribution. They do not contribute to any distinguishing advantage. Removing them we get:
\begin{align}
\mathit{Real:} &\ket{+_{\theta_s}}\ket{\theta}\ket{m_x}\ket{(-1)^{m_x}\theta - \theta_s}, \\
\mathit{Sim:} &\ket{+_{(-1)^{m_{s, x}}(\theta+\theta_s) + m_{s, Z}\pi}}\ket{\theta}\ket{m_{s, x}}\ket{-(-1)^{m_{s,x}}\theta_s - m_{s,z}\pi}.
\end{align}
The distinguisher then applies to the qubit state a $\RZ$ rotation whose angle is given by the last value in the equations above:
\begin{align}
\mathit{Real:} &\ket{+_{(-1)^{m_x}\theta}}\ket{\theta}\ket{m_x}\ket{(-1)^{m_x}\theta - \theta_s}, \\
\mathit{Sim:} &\ket{+_{(-1)^{m_{s, x}}\theta}}\ket{\theta}\ket{m_{s, x}}\ket{-(-1)^{m_{s,x}}\theta_s - m_{s,z}\pi}.
\end{align}
We note that $\theta_s$ only appears in the last values now and these values in both cases are sampled from the same distribution since $\theta_s$ is perfectly random. The correction bits $m_x$ and $m_{s, x}$ are both sampled uniformly at random (the first directly by Alice and the other via the teleportation procedure). These distributions are the same and therefore perfectly indistinguishable, which concludes the proof.

\end{proof}

\begin{proof}[Combining security and correctness to prove Theorem~\ref{thm:sec-gadget}]
For $t \leq n \eta_1$, the protocol is $\epsilon_{cor}$-correct, while for $t \geq n p_{\alpha, 2}$, the protocol is $\epsilon_{sec}$-secure. For $p_{\alpha, 2} \leq \eta_1$ we can combine these two conditions by choosing $p_{\alpha, 2} \leq t/n \leq \eta_1$. 
The overall optimal AC security bound of the protocol is given by:
\begin{align}
\epsilon &= \min_{p_{\alpha, 2} \leq t/n \leq \eta_1} (\max(\epsilon_{cor}, \epsilon_{sec})) \\
&= \min_{p_{\alpha, 2} \leq t/n \leq \eta_1} \left(\max\left(\exp \left( -2 \left(\eta_1 - t/n \right)^2 n \right), \exp \left( -2 \left(t/n - p_{\alpha, 2} \right)^2 n \right)\right)\right).
\end{align}
At equilibrium both values are equal, which implies:
\begin{equation}
t/n = \frac{\eta_1 + p_{\alpha, 2}}{2}.
\end{equation}
This yields the desired result when injected in the expression for $\epsilon$.

\end{proof}

\section{Post-selected protocol for privacy amplification}
\label{app:post-select-prot}

In this section we present an alternative protocol in which the abort is triggered as soon as a single photon is missing during detection. This increases the security to the detriment of the success probability.

\begin{protocol}[ht]
\caption{GHZ Privacy Amplification for Rotated States from Weak Coherent Pulses with Post-Selection}
\label{prot:post-select-ghz-gadget}
\begin{algorithmic} [0]

\STATE \textbf{Public Information:} Laser pulse intensity $\abs{\alpha}^2$, number $n$ of qubits in the GHZ gadget.

\STATE \textbf{Inputs:} Alice inputs an angle $\theta \in \Phi$. Bob inputs a single-qubit spin state $\rho_{qe}$.

\STATE \textbf{The Protocol:}

\begin{enumerate}
\item Alice samples values $(\theta_1, \ldots, \theta_{n-1})$ uniformly at random from $\Phi$, together with a uniformly random bit $m_x$, and sets $\theta_n = (-1)^{m_x}\theta - \sum_{i=1}^{n-1}\theta_i$.
\item For $1 \leq i \leq n$:
\begin{enumerate}
\item Alice sends the phase-randomised rotated weak coherent pulse $\rho_{\alpha, \theta_i}$ to Bob.
\item Bob uses this pulse to apply the rotated emission operator $E_{\rm qe}(\theta_i)$ to his spin qubit.
\end{enumerate}
\item Finalisation:
\begin{enumerate}
\item Alice sends the bit $m_x$ to Bob.
\item Bob emits another photon using the emission operator $E_{\rm qe}$. This qubit is indexed $0$.
\item Bob attempts to measures all photonic qubits $i \neq 0$ in the GHZ state in the $\ket{\pm}$ basis. If there is a photon loss, signalled by a measurement not returning a result, Bob sends $\Ab$ to Alice.
\item Otherwise let $b = \bigoplus_{i = 1}^{n} b_i$ be the parity of the measurement outcomes. Bob applies $\Z^b$ to the remaining photonic qubit $0$ and sets it as his output.
\end{enumerate}
\end{enumerate}

\end{algorithmic}
\end{protocol}

The security guarantees of Protocol~\ref{prot:post-select-ghz-gadget} are given in Theorem~\ref{thm:post-select-sec-gadget} below.

\begin{theorem}[AC Security of Protocol~\ref{prot:post-select-ghz-gadget}]
\label{thm:post-select-sec-gadget}

Let $\eta_1$ be the probability of the the quantum emitter generates a photon after receiving a laser pulse of intensity $\abs{\alpha}^2$ and let $p_{\alpha, 2} = 1 - e^{-\abs{\alpha}^2} - \abs{\alpha}^2 e^{-\abs{\alpha}^2}$ be the probability that the laser pulse contains two or more photons. Let $\alpha, n$ be the parameters used in Protocol~\ref{prot:post-select-ghz-gadget}. Then Protocol~\ref{prot:post-select-ghz-gadget} $\epsilon'$-constructs in the Abstract Cryptography framework the Blind Graph State Extender Resource~\ref{res:graph-extend}, for $\epsilon' = \max(1-\eta_1^n, p_{\alpha, 2}^n)$.

\end{theorem}

\begin{proof}[Proof of correctness]

In the case where Bob is honest and does not abort, the correctness analysis is the same as in the correctness proof of Protocol~\ref{prot:ghz-gadget}. For a pure input state $\ket{\psi}_{qe} = \alpha\ket{0} + \beta\ket{1}$, we get after the second step:
\begin{equation}
\RZ\left(\sum_{i = 1}^n \theta_i\right)\left(\alpha\ket{0}^{n+2} + \beta\ket{1}^{n+2}\right) = \RZ((-1)^{m_x}\theta)\left(\alpha\ket{0}^{n+2} + \beta\ket{1}^{n+2}\right),
\end{equation}
which becomes $\alpha\ket{0}_{qe}\ket{0}_{ph} + (-1)^b e^{i(-1)^{m_x}\theta}\beta\ket{1}_{qe}\ket{1}_{ph}$ after Bob's measurements in the $\ket{\pm}$ basis. After the correction the state is exactly correct for the associated bit $m_x$.

We now analyse the abort probability. The protocol aborts if and only if at least one photon is not emitted. 
The random variable counting the number of emitted photons follows a binomial distribution of parameter $(\eta_1, n)$. Therefore, the abort probability and correctness error are:
\begin{equation}
\epsilon_{cor} = \Pr[\Ab] = 1 - \eta_1^n.
\end{equation}

\end{proof}

\begin{proof}[Proof of security]

To prove the security against malicious Bob, we must construct a simulator which has a single oracle access to the RSP resource and must emulate the communication of Alice without knowing which state Alice has chosen to send. We will then show that this exchange is indistinguishable to any unbounded distinguisher from a communication with an honest Alice, up to a known failure probability.

The simulator is described in Simulator~\ref{sim:mal-bob} below.

\begin{simulator}[ht]
\caption{Malicious Bob}
\label{sim:post-select-mal-bob}

\begin{enumerate}
\item The simulator samples $n$ values $\{k_i\}_{1 \leq i \leq n}$ from the Poisson distribution with parameter $\abs{\alpha}^2$.
\item If there is no index $i$ such that $k_i \in \bin$, the simulator outputs $\err$ to the distinguisher and stops. Otherwise it continues.
\item Let $s$ be one index with $k_s = 1$, chosen uniformly at random from suitable values. If there is no value of this sort, it chooses similarly $s$ such that $k_s = 0$.
\item The simulator performs the call to the RSP Resource~\ref{res:rsp} and receives a state $\ket{+_\theta}$ for an unknown value of $\theta$. It samples a uniformly random bit $m_{s, x}$ and applies $\X^{m_{s, x}}$ to this qubit.
\item It samples values $(\theta_1, \ldots, \theta_{n-1})$ uniformly at random from $\Phi$ and sets $\theta_n = - \sum_{i=1}^{n-1}\theta_i$.
\item For each pulse $i$ which Alice sends to Bob in the real protocol:
\begin{itemize}
\item If $i \neq s$, the simulator sends $\ket{k_i}_{\theta_i}$;
\item If $i = s$ and $k_s = 1$, it applies $\RZ(\theta_s)$ to the state received from the RSP resource and sends the resulting state. If $k_s = 0$, it discards this state and sends nothing.
\end{itemize}
\item The simulator then sends $m_{s, x}$ and stops.
\end{enumerate}

\end{simulator}

It is clear that in the case where the simulator aborts by sending $\err$, the states are perfectly distinguishable since an honest Alice never sends this message in a real execution of the protocol. This is the case if all pulses contain at least two photons. Recall that $p_{\alpha, 2} = 1 - (1 + \abs{\alpha}^2) e^{-\abs{\alpha}^2}$ is the probability that a given pulse contains two or more photon. The random variable counting the number of pulses with more than one photon follows a binomial distribution of parameters $(p_{\alpha, 2}, n)$. The simulator sends the message $\err$ with probability $p_{\alpha, 2}^n$. We now show that if there is no abort, then the real and ideal states follow exactly the same distribution, meaning that no distinguisher can tell apart the simulation and the real execution.

In the protocol, Alice sends the states:
\begin{equation}
    \rho_{\alpha,\theta_i} = e^{- \abs{\alpha}^2} \sum_{k=0}^\infty \frac{\abs{\alpha}^{2k}}{k!} \op{k}_{\theta_i}.
\end{equation}
The distinguisher can perform a non-destructive photon counting measurement to know how many photons it has received in each pulse. In the real case, the distinguisher then has in its possession the following states, for known values of $k_i$:
\begin{equation}
\ket{k_s}_{\theta_s}\bigotimes_{\substack{i = 1 \\ i \neq s}}^n \ket{k_i}_{\theta_i}.
\end{equation}

In the simulated case, the distinguisher receives:
\begin{equation}
\ket{k_s}_{(-1)^{m_{s, x}}\theta + \theta_s}\bigotimes_{\substack{i = 1 \\ i \neq s}}^n \ket{k_i}_{\theta_i}.
\end{equation}

We need to show that these states are indistinguishable, even after receiving $m_x$ or $m_{s, x}$ from Alice or the simulator respectively. The values $k_i$ follows the same distribution in the real and simulated cases. In the worst case there is a single value $k_s$ that satisfies the condition for continuing the simulation, meaning that for all $i \neq s$ we have $k_i \geq 2$. We can also assume that for all these other indices, the entire value of $\theta_i$ is leaked to the distinguisher. We can then replace the states $\bigotimes_{\substack{i = 1 \\ i \neq s}}^n\ket{k_i}_{\theta_i}$ with $\bigotimes_{\substack{i = 1 \\ i \neq s}}^n\ket{\theta_i}$, where the $\ket{\theta_i}$ are classical states.

In the real protocol we have $\theta_s = (-1)^{m_x}\theta - \sum_{\substack{i = 1 \\ i \neq s}}^n \theta_i$, while in the simulation we have $\theta_s = - \sum_{\substack{i = 1 \\ i \neq s}}^n \theta_i$, with the other values being sampled uniformly at random. Replacing these values above we get:
\begin{align}
\mathit{Real:} &\ket{k_s}_{(-1)^{m_x}\theta - \sum_{\substack{i = 1 \\ i \neq s}}^n \theta_i}\bigotimes_{\substack{i = 1 \\ i \neq s}}^n \ket{\theta_i}, \\
\mathit{Sim:} &\ket{k_s}_{(-1)^{m_{s, x}}\theta - \sum_{\substack{i = 1 \\ i \neq s}}^n \theta_i}\bigotimes_{\substack{i = 1 \\ i \neq s}}^n \ket{\theta_i},
\end{align}
for $k_s \in \bin$. These distributions are the same and therefore perfectly indistinguishable, which concludes the proof.

\end{proof}

\end{document}